\documentclass[acmsmall, screen, nonacm]{acmart}
\usepackage{ltl}
\usepackage{dashbox}
\usepackage{mathpartir}
\usepackage{csquotes}
\usepackage{hyperref}
\usepackage{amsthm}
\usepackage{thmtools}
\usepackage{nameref, cleveref}
\usetikzlibrary{automata, positioning}

\theoremstyle{remark}
\newtheorem{example}{Example}

\Crefname{theorem}{Theorem}{Theorems}
\Crefname{lemma}{Lemma}{Lemmas}
\Crefname{example}{Example}{Examples}

\AtBeginEnvironment{example}{%
  \pushQED{\qed}%
}
\AtEndEnvironment{example}{\popQED\endexample}



\begin{abstract}
  It is well known that liveness properties cannot be proven using standard simulation arguments. This issue has been mitigated by extending standard notions of simulation for transition systems to \emph{fairness-preserving simulations} for systems equipped with an additional fairness condition modeling liveness assumptions and/or liveness requirements.
  In the context of automated verification of finite-state systems, proofs by simulation are an appealing method as there exist efficient algorithms to find a simulation between two systems.
  However, applications of fair simulation to interactive verification have been much less studied.
  Perhaps one reason is that the definitions of fair simulation relations typically involve non-trivial nestings of inductive and coinductive relations, making them particularly difficult to use and to reason about.
  In this paper, we argue that in many cases, stronger notions of fair simulation involving more controlled alternations of fixed points are sufficient.
  Starting from known fair simulation techniques, we progressively build up a family of \emph{almost fair} simulation relations for transition systems equipped with a B\"uchi fairness condition.
  The simulation relations we present can all be equipped with intuitive reasoning rules, leading to elegant deductive systems to prove fair trace inclusion.
  We mechanized our simulation relations and their associated deductive systems in the Rocq proof assistant, proved their soundness, and we demonstrate their use through a selection of examples.
\end{abstract}

\newcommand{\iostep}[3]{#1 \xrightarrow{#2} #3}

\newcommand{\lang}{\mathcal{L}}

\newcommand*{\cbox}[1]{\fbox{\ensuremath{#1}}}

\newcommand*{\obox}[1]{\dbox{\ensuremath{#1}}}

\newcommand{\triple}[4]{\obox{#1} \vdash #2 #3 #4}
\newcommand{\gtriple}[4]{\cbox{#1} \vdash #2 #3 #4}

\newcommand{\tripleDelay}[4]{
  \obox{#1} \vdash_\texttt{d} #2 #3 #4
}
\newcommand{\gtripleDelay}[4]{
  \cbox{#1} \vdash_\texttt{d} #2 #3 #4
}

\newcommand{\tripleRdelay}[4]{
  \obox{#1} \vdash_\texttt{rd} #2 #3 #4
}
\newcommand{\gtripleRdelay}[4]{
  \cbox{#1} \vdash_\texttt{rd} #2 #3 #4
}

\newcommand{\triplerb}[4]{
  \obox{#1} \vdash_\texttt{rb} #2 #3 #4
}
\newcommand{\gtriplerb}[4]{
  \cbox{#1} \vdash_\texttt{rb} #2 #3 #4
}

\newcommand{\relL}{\preccurlyeq_\texttt{L}}
\newcommand{\relR}{\preccurlyeq_\texttt{R}}
\newcommand{\relW}{\preccurlyeq_\texttt{W}}

\newcommand{\fsimdirect}{\texttt{fsim}_\texttt{direct}}
\newcommand{\fsimdirectF}{\texttt{fsimF}_\texttt{direct}}

\newcommand{\fsimdelay}{\texttt{fsim}_\texttt{delay}}
\newcommand{\fsimdelayR}{\texttt{fsim}^\texttt{R}_\texttt{delay}}
\newcommand{\fsimdelayL}{\texttt{fsim}^\texttt{L}_\texttt{delay}}

\newcommand{\fsimdelayFL}{\texttt{fsimF}^\texttt{L}_\texttt{delay}}

\newcommand{\ddelaysim}{\texttt{fsim}_\texttt{2delay}}

\newcommand{\delay}{\texttt{wait}} 

\newcommand{\rdelaysim}{\texttt{fsim}_\texttt{rdelay}}
\newcommand{\rdelaysimW}{\texttt{fsim}^\texttt{W}_\texttt{rdelay}}
\newcommand{\rdelaysimL}{\texttt{fsim}^\texttt{L}_\texttt{rdelay}}
\newcommand{\rdelaysimR}{\texttt{fsim}^\texttt{R}_\texttt{rdelay}}

\begin{document}

\title{Almost Fair Simulations}

\author{Arthur Correnson}
\orcid{0000-0003-2307-2296}
\affiliation{%
  \institution{CISPA Helmholtz Center for Information Security}
  \city{Saarbrücken}
  \country{Germany}
}
\email{arthur.correnson@cispa.de}

\author{Iona Kuhn}
\orcid{0009-0008-4109-6092}
\affiliation{%
  \institution{Saarland University}
  \city{Saarbrücken}
  \country{Germany}
}
\email{ioku00001@stud.uni-saarland.de}

\author{Bernd Finkbeiner}
\orcid{0000-0002-4280-8441}
\affiliation{%
  \institution{CISPA Helmholtz Center for Information Security}
  \city{Saarbrücken}
  \country{Germany}
}
\email{finkbeiner@cispa.de}

  \maketitle

  \section{Introduction}
  Simulation proofs provide a systematic technique to reduce proofs of trace inclusion between programs into simple \textit{local} reasoning about states and transitions.
  In the context of program verification, the problem of checking whether
  a program satisfies a given safety specification (e.g. specifications stating that \textit{"nothing bad ever happens"}) can always be reformulated as a trace
  inclusion between a source transition system modeling the program and a target transition system that nondeterministically produces safe behaviors.
  In turn, simulation can immediately be applied as a simple technique to prove safety properties.
  




  Unfortunately, this technique does not immediately work for liveness specifications (e.g. specifications stating that \textit{"something good should eventually happen"}) such as termination or response properties. Indeed, modeling liveness specifications (and also programs with liveness assumptions) as transition systems typically requires augmenting the transition systems with an additional \emph{fairness condition} describing which executions are legitimate, and which are not.
  In this context, proving trace inclusion means proving that any 
  trace engendered by a fair execution of the source transition system can be reproduced via a fair execution of the target transition system.
  Standard simulation techniques would only ensure that any execution of the source can be mimicked by an arbitrary execution of the target, disregarding the fairness conditions.
  Instead, we need a notion of simulation that filters out unfair executions of the source, as these do not need to be simulated, and further enforces that the remaining fair executions of the source can be simulated by \emph{fair} executions of the target.

  It is important to note that fairness conditions are typically \textit{global} conditions inspecting an entire execution to determine whether it is fair or not.
  This contrasts with one of the main benefits of simulation techniques, which is \emph{precisely} to reduce global reasoning about executions to local reasoning about states and transitions. In turn, designing fairness-preserving simulation relations is non-trivial and it requires to carefully approximate the global fairness conditions with more local checks on states and transitions.

  To this extent, several extended notion of simulation such as direct simulation, delay simulation, and fair simulation have been proposed in the literature \cite{equivalence-fair-kripke, fair-simulation,delay-simulation}.
  However, the focus has been mostly on the design of efficient algorithms to automatically construct simulations between two finite-state transition systems, or on algorithms to minimize large systems by computing their quotient modulo a fairness-preserving simulation relation \cite{smaller-quotients}.
  The use of fairness-preserving simulation relations for interactive verification, and in particular interactive verification of liveness properties, has been comparatively much less explored.

  In this paper, we propose to re-explore already established notions of fairness-preserving simulation, but from the point of view of interactive deductive verification inside a proof assistant. More precisely, we focus on fairness-preserving simulation of B\"uchi automata, an expressive computational model 
  that is commonly used in the literature on automated verification of linear time properties. We present the following contributions: \begin{enumerate}
    \setlength\itemsep{1em}
    \item Our first contribution is to formalize direct simulation and delay simulation of B\"uchi automata in the Rocq proof assistant. While both notions are well-established in the literature, their rigorous formalization in a proof assistant
    turns out to be fairly technical.
    Further, by leveraging the framework of \emph{parameterized coinduction}, we demonstrate that both direct and delay simulation can be presented in the form of a deductive system to prove language inclusion of B\"uchi automata interactively.
    \item Through a selection of examples, we precisely identify limitations of delay simulation and argue that for the purpose of interactive proofs of liveness properties, weaker notions of fairness-preserving simulations are required.
    We nonetheless observe that in the specific case where the left-hand automaton is a safety automaton (i.e., when we only have liveness requirements and no liveness assumptions),
    a simpler notion of \emph{right-biased} delay simulation can be sufficient.
    \item To mitigate the limitations of direct and delay simulation in the case where
    the left-hand automaton is \emph{not} a safety automaton (i.e., when we have both liveness assumptions and liveness requirements),
    we present two novel notions of fairness-preserving simulation: \emph{double delay simulation}, and \emph{repeated delay simulation}.
    Both are significantly weaker than delay simulation, but they remain sound for language inclusion of B\"uchi automata. Further, they preserve the simplicity and ease-of-use of delay simulation and can also be equipped with an intuitive set of reasoning rules. We mechanized these new notions of simulation and their proof of soundness in Rocq.
  \end{enumerate}

  \section{Preliminaries}
  \label{sec:preliminaries}

  Before proceeding with a description of the contributions, we start by introducing necessary background on labeled transition systems, B\"uchi automata, and parameterized coinduction, which all play a central role throughout the paper.

  \subsection{Labeled Transition Systems}

  A labeled transition system (LTS) is a triple $(\mathcal{S}, \mathcal{E}, \mathcal{I}, \to)$ where $\mathcal{S}$ is a set of states, $\mathcal{E}$ is a set of events, $\mathcal{I} \subseteq \mathcal{S}$ is a set of initial states, and $\mathcal{\to} \subseteq \mathcal{S} \times \mathcal{E} \times \mathcal{S}$ is a labeled transition relation.
  $s_1 \xrightarrow{e} s_2$ indicates that it is possible to transition from $s_1$ to $s_2$ while emitting/reading the event $e$. Given a LTS $\mathit{TS} = (\mathcal{S}, \mathcal{E}, \mathcal{I}, \to)$, the set of \emph{traces} that can be produced starting from a given state $s \in \mathcal{S}$ is the set of infinite sequences of events $\tau \in \mathcal{E}^\omega$ characterized by the following coinductive relation $\mathit{Traces}_\mathit{TS} \subseteq \mathcal{S} \times \mathcal{E}^\omega$.

  \begin{definition}[Traces] $\mathit{Traces}_\mathit{TS} \triangleq \nu T . \{ \ (s, e \cdot \tau) \mid \exists s \xrightarrow{e} s'. \ (s', \tau) \in T \ \}$
  \end{definition}

  For a specific state $s \in \mathcal{S}$ we note $\mathit{Traces}_\mathit{TS}(s) \triangleq \{ \tau \mid (s, \tau) \in \mathit{Traces}_\mathit{TS} \}$.
  For a set of states $S \subseteq \mathcal{S}$ we note $\mathit{Traces}_\mathit{TS}(S) \triangleq \bigcup_{s \in S} \mathit{Traces}_\mathit{TS}(s)$.
  We will often consider the set $\mathit{Traces}_\mathit{TS}(\mathcal{I})$ of all \emph{initial} traces of $\mathit{TS}$.
  Alternatively, we note $\mathit{Traces}(\mathit{TS}) \triangleq \mathit{Traces}_\mathit{TS}(\mathcal{I})$.
  When clear from context, we omit the subscript $\mathit{TS}$.

  \subsection{B\"uchi Automata}

  A B\"uchi automaton $(\mathcal{S}, \mathcal{E}, \mathcal{I}, \mathcal{F}, \to)$
  is a LTS extended with a set of accepting states $\mathcal{F} \subseteq \mathcal{S}$ that should be visited infinitely often over the course of an execution for its associated trace to be considered valid. Traces that can be produced via infinitely many visits to accepting states are said to be in the language of the automaton.
  Formally, the language of a state $s$ of an automaton A is characterized by the following
  coinductive-inductive relation $\lang_\mathit{A} \subseteq \mathcal{S} \times \mathcal{E}^\omega$.

  \begin{definition}[Language]
    \begin{align*}
      \lang_\mathit{A} \triangleq \nu L . \mu X . \ &\{ \ (s, e \cdot \tau) \mid \exists s \xrightarrow{e} s'. \ (s', \tau) \in X \ \} \ \cup\\
      &\{ \ (s, e \cdot \tau) \mid s \in \mathcal{F} \wedge \exists s \xrightarrow{e} s'. \ (s', \tau) \in L \ \}
    \end{align*}
  \end{definition}

  As for traces, we use the notations $\lang_A(s) \triangleq \{ \tau \mid (s, \tau) \in \lang_A \}$, $\lang_A(S) \triangleq \bigcup_{s \in S} \lang_A(s)$, and $\lang(A) \triangleq \lang_A(\mathcal{I})$.

  Contrary to LTSs, which can only model safety requirements, B\"uchi automata can also model specifications with a liveness component.
  For example, the following automaton over events $\mathcal{E} = \{ a, b \}$ expresses the requirement that an event "a" must eventually be emitted.

  
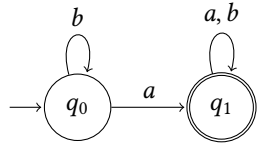
\begin{center}
    \begin{tikzpicture}[shorten >=1pt, node distance=1cm, auto, initial text=]
      \node[initial, state] (A) {$q_0$};
      \node[accepting, state, right=of A] (B) {$q_1$};
      \path[->]
        (A) edge node{$a$} (B)
        (A) edge[loop above] node{$b$} (A)
        (B) edge[loop above] node{$a, b$} (B);
    \end{tikzpicture}
  \end{center}

  In this automaton, the only accepting state is marked with a double circle.
  Viewed as a LTS, this automaton accepts any infinite sequence of events $a$ and $b$. However, the addition of the accepting set $\mathcal{F} = \{ q_1 \}$ effectively filters out traces containing only $b$'s.

  We note that B\"uchi automata are usually assumed to have finite state spaces (i.e., $|\mathcal{S}| < \infty$).
  This limits the specifications that can be encoded to so called \emph{omega-regular properties} \cite{Baier2008PrinciplesOM}.
  The advantage of considering only finite-state automata is that it gives access to efficient decision procedures to solve problems such as language inclusion.
  This observation is the basis of 
  automata-based model-checking Linear Temporal Logics.
  In this paper, we are not concerned with automated verification. Instead, we aim to develop interactive deductive proof techniques. In this setting, assuming the finiteness of the state-space is an artificial limitation.

  \subsection{(Parameterized) Coinduction}
  \label{sec:paco}

  Let $A$ be a set, and $(\mathcal{P}(A), \subseteq, \cap, \cup, A, \emptyset)$ be the complete lattice of subsets of $A$ (ordered by set inclusion). A result due to Tarski \cite{tarski} guarantees that any monotone functor $F : \mathcal{P}(A) \xrightarrow{\mathit{mon}} \mathcal{P}(A)$ has a greatest fixed point noted $\nu F$. Further, $\nu F$ is exactly the union of all postfixed points of $F$: \[
    \nu F \triangleq \bigcup \ \{ \ R \mid R \subseteq F(R) \ \}
  \]

  An immediate consequence of this result is that any coinductive predicate $\nu F \subseteq A$ has a systematic proof technique associated with it. Indeed, let $X \subseteq A$ and suppose we want to prove that $X \subseteq \nu F$.
  Since $\nu F$ is larger than any postfixed point of $F$, it suffices to exhibit a postfixed point of $R$ of $F$ containing $X$.
  In this context, $R$ is usually referred to as a \emph{coinduction hypothesis}.
  
  \begin{lemma}[Coinduction Principle]
    \label{coind}
    $X \subseteq \nu F \iff \exists R, X \subseteq R \wedge R \subseteq F(R)$
  \end{lemma}

  A common use case for proofs by coinduction are proofs of trace inclusion by simulation.
  Given two LTSs $\mathit{TS}_1 = (\mathcal{S}_1, \mathcal{E}, \mathcal{I}_1, \to_1)$ and $\mathit{TS}_2 = (\mathcal{S}_2, \mathcal{I}_2, \mathcal{E}, \to_2)$ with the same set of events, to prove that they are \emph{trace included} (i.e., $\mathit{Traces}(\mathcal{I}_1) \subseteq \mathit{Traces}(\mathcal{I}_2)$) it is well known that it suffices to show that $\mathcal{I}_1 \times \mathcal{I}_2 \subseteq \texttt{sim}$ where $\texttt{sim}$ is the binary relation defined coinductively as follows. 

  \begin{definition}[Standard Simulation]
    \begin{align*}
      \texttt{simF}(X) &\triangleq \{ \ (s_1, s_2) \mid \forall e. \forall \iostep{s_1}{e}{s'_1}. \exists \iostep{s_2}{e}{s'_2} \wedge (s'_1, s'_2) \in X \ \}\\
      \texttt{sim} &\triangleq \nu \texttt{simF}
    \end{align*}
  \end{definition}

  \begin{example}
    \label{ex:simple-sim}

    As an example, consider the two following LTSs labeled with events $\mathcal{E} = \{ a, b, \epsilon \}$. The left-hand one could represent a program repeatedly emitting an event "a" and the right-hand one a safety specification requiring that only event "a" or event "b" is ever emitted.

  
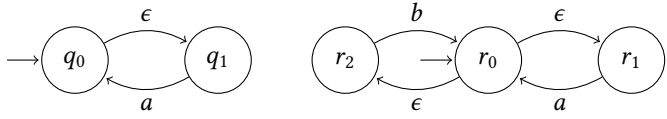
\begin{center}
    \begin{tikzpicture}[shorten >=1pt, node distance=1cm, auto, initial text=]
      \node[initial, state] (A) {$q_0$};
      \node[state, right=of A] (B) {$q_1$};
      \path[->]
        (A) edge[bend left] node {$\epsilon$} (B)
        (B) edge[bend left] node {$a$} (A);
    \end{tikzpicture}
    \qquad
    \begin{tikzpicture}[shorten >=1pt, node distance=1cm, auto, initial text=]
      \node[initial, state] (A) {$r_0$};
      \node[state, right=of A] (B) {$r_1$};
      \node[state, left=of A] (C) {$r_2$};
      \path[->]
        (A) edge[bend left] node {$\epsilon$} (B)
        (B) edge[bend left] node {$a$} (A)
        (A) edge[bend left] node {$\epsilon$} (C)
        (C) edge[bend left] node[above] {$b$} (A);
    \end{tikzpicture}
  \end{center}

  It is not difficult to check that $R = \{ (q_0, r_0), (q_1, r_1) \}$ contains $\mathcal{I}_1 \times \mathcal{I}_2 = \{ (q_0, r_0) \}$ and is a postfixed point of $\texttt{simF}$.
  Therefore, by \Cref{coind}, $\mathcal{I}_1 \times \mathcal{I}_2 \subseteq \texttt{sim}$, and the left-hand system satisfies the right-hand specification.
  \end{example}

  For this simple example, guessing a postfixed point up-front was not difficult. For larger systems, this can rapidly become more challenging. Starting with at least $\mathcal{I}_1 \times \mathcal{I}_2$ is necessary, but in practice it is not sufficient (in the above example, $\mathcal{I}_1 \times \mathcal{I}_2$ is not a postfixed point!), requiring to restart the proof with a slightly larger guess until we eventually find a large enough one.
  In the context of automated verification, this trial and error process to prove simulation is not necessarily a limitation, and intermediate results can be automatically cached and re-exploited. However, in the context of interactive deductive verification, in particular inside a proof assistant, it is cumbersome.
  In \cite{paco}, it has been observed that even in an interactive setting, there is actually no strict need to guess the coinduction hypothesis up-front.
  Better, there is actually no need to ever provide a postfixed point
  at all and it generally suffices to collect fragments of a postfixed point instead~\cite{paco, gpaco}.
  The key idea supporting this technique is to replace the usual fixed point operator $\nu$ with a parameterized version of it, $G_F(H)$, where $H$ is a current \emph{guess} for (a fragment of) the coinduction hypothesis. Formally, $G$ is defined as follows: \[
    G_F(H) \triangleq \nu X. F(X \cup H)
  \]
  It is not difficult to see that $G$ coincides with the standard greatest fixed point operator when $H = \emptyset$ (i.e., $G_F(\emptyset) = \nu F$).
  In particular, it means that proving $X \subseteq G_F(\emptyset)$ is equivalent to proving $X \subseteq \nu F$.
  The benefit is that the parameterized greatest fixed point admits the following \emph{incremental} reasoning rules: \begin{mathpar}
    \inferrule[Init]{X \subseteq G_F(\emptyset)}{X \subseteq \nu . F}
    \qquad
    \inferrule[Accumulate]{X \subseteq G_F(H \cup X)}{X \subseteq G_F(H)}
    \qquad
    \inferrule[Step]{X \subseteq F(H \cup G_F(H))}{X \subseteq G_F(H)}
    \label{Paco}
  \end{mathpar}

  The rule \textsc{Init} initializes a proof by \emph{parameterized coinduction} by replacing a goal of the form $X \subseteq \nu F$ with the equivalent $X \subseteq G_F(\emptyset)$.
  The rule \textsc{Accumulate} extends the current guess $H$ with the subset $X$ whose inclusion in the greatest fixed point should be proven.
  Finally, the rule \textsc{Step} enables to make progress by unfolding the greatest fixed point underlying the definition of $G$.
  We note that by definition of $G$, every element previously accumulated in $H$ can be used after applying the rule \textsc{Step}.
  To better illustrate how to exploit the rule of parameterized coinduction, we revisit the previous example and present a detailed incremental proof that $\mathcal{I}_1 \times \mathcal{I}_2 \subseteq \texttt{sim}$.
  
  \begin{example}[\Cref{ex:simple-sim}, revisited]
    \begin{align*}
      &\{ (q_0, r_0) \} \in \texttt{sim}\\
      \textrm{\color{gray}By \textsc{Init}}
      \impliedby &\{ (q_0, r_0) \} \in G_\texttt{simF}(\emptyset)\\
      \textrm{\color{gray}By \textsc{Accumulate}}
      \impliedby &\{ (q_0, r_0) \} \in G_\texttt{simF}(\{ q_0, r_0 \})\\
      \textrm{\color{gray}By \textsc{Step}}
      \impliedby &\exists \iostep{r_0}{\epsilon}{r}. (q_1, r) \in \{ q_0, r_0 \} \cup G_\texttt{simF}(\{ q_0, r_0 \})\\
      \textrm{\color{gray}Pick $r = r_1$}
      \impliedby &(q_1, r_1) \in G_\texttt{simF}(\{ q_0, r_0 \})\\
      \textrm{\color{gray}By \textsc{Step}}
      \impliedby &\exists \iostep{r_1}{a}{r}. (q_0, r) \in \{ q_0, r_0 \} \cup G_\texttt{simF}(\{ q_0, r_0 \})\\
      \textrm{\color{gray}Pick $r = r_0$}
      \impliedby &(q_0, r_0) \in \{ q_0, r_0 \} \cup G_\texttt{simF}(\{ q_0, r_0 \})\\
      \impliedby &(q_0, r_0) \in \{ q_0, r_0 \}
    \end{align*}
  \end{example}

  We note that in this example, even though $\{ (q_0, r_0) \}$ is not a postfixed point of $\texttt{simF}$, we never needed to accumulate more that just the initial states!
  We however note that it is sometimes necessary to accumulate more than just the currently targeted subset $X$. For this purpose, the following \textsc{Invariant} rule can easily be derived as a corollary of \textsc{Accumulate}. \begin{mathpar}
    \inferrule[Invariant]{X \subseteq H' \\ H' \subseteq G_F(H' \cup H)}{X \subseteq G_F(H)}
  \end{mathpar}
  \begin{proof}
    Suppose (1) $X \subseteq H'$, and (2) $H' \subseteq G_F(H' \cup H)$. By \textsc{Accumulate} and (2) we obtain $H' \subseteq G_F(H)$. By (1) it follows that $X \subseteq H' \subseteq G_F(H)$.
  \end{proof}

  \section{Direct Simulation}

  In the previous section, we discussed the use of simulation combined with parameterized coinduction as an interactive proof technique to show trace inclusion between two LTSs.
  A natural question to ask is whether the same incremental proof technique can also be used to prove \emph{language inclusion} of LTSs equipped with a fairness condition such as B\"uchi automata.
  Clearly, the standard notion of simulation is not suitable to check inclusion of B\"uchi automata, as it does not even mention accepting sets. For example, consider the following two automata:

  
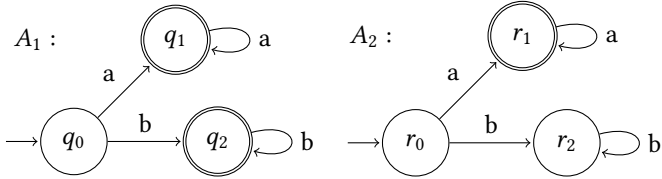
\begin{center}
    \begin{tikzpicture}[shorten >=1pt, node distance=1cm, auto, initial text=]
      \node[state,initial] (q_0) {$q_0$};
      \node[state,accepting] (q_1) [above right=of q_0] {$q_1$};
      \node[state,accepting] (q_2) [right=of q_0] {$q_2$};
      \node[left=of q_1] (A1) {$A_1:$};
      \path[->] 
        (q_0) edge node {b} (q_2)
        (q_0) edge node {a} (q_1)
        (q_1) edge [loop right] node {a} ()
        (q_2) edge [loop right] node {b} ();
    \end{tikzpicture}
    \begin{tikzpicture}[shorten >=1pt,node distance=2cm,on grid,auto,initial text=]
      \node[state,initial] (q_0) {$r_0$};
      \node[state,accepting] (q_1) [above right=of q_0] {$r_1$};
      \node[state] (q_2) [right=of q_0] {$r_2$};
      \node[left =of q_1] (A2) {$A_2:$};
      \path[->] 
        (q_0) edge node {a} (q_1)
        (q_0) edge node {b} (q_2)
        (q_1) edge [loop right] node {a} ()
        (q_2) edge [loop right] node {b} ();
    \end{tikzpicture}
  \end{center}

  Both automata have the same traces and they are even (bi)similar. However, they do not have the same languages! Indeed, the word $b^\omega$ is accepted by $A_1$ but rejected by $A_2$ as the only corresponding execution for $b^\omega$ loops in the state $r_2$, which is not an accepting state.
  To overcome this issue, a naive extension of standard simulation would be to synchronize visits to accepting states.
  In other words, whenever the left-state is accepting, we require the right-state to be accepting as well.
  This extension is usually referred to as \emph{direct simulation}, and can be defined coinductively as follows:\begin{align*}
    \fsimdirect &\triangleq \nu X . \{ \ (s_1, s_2) \mid (s_1 \in \mathcal{F}_1 \implies s_2 \in \mathcal{F}_2) \wedge \forall e. \forall \iostep{s_1}{e}{s'_1}. \exists \iostep{s_2}{e}{s'_2}. \ (s'_1, s'_2) \in X \ \}
  \end{align*}

  Since all visits to a left-accepting state are exactly synchronized with a visit to a right-accepting state, it is not too difficult to see that $\fsimdirect$ offers a sound proof technique for language inclusion.

  \begin{theorem}[Soundness]
    \label{thm:fsimdirect-sound}
    $(s_1, s_2) \in \fsimdirect \implies \lang(s_1) \subseteq \lang(s_2)$
  \end{theorem}

  \begin{proof}
    The proof goes by coinduction on $\lang(s_2)$ (remember that $\lang$ is a coinductive predicate)
    and then by induction on the inductive part of $\lang(s_1)$.
  \end{proof}

  Contrary to standard simulation, direction simulation (rightfully) rejects the two automata from the previous example.
  However, it can prove language inclusion of the previous example in the other direction (i.e., $\lang(A_2) \subseteq \lang(A_1)$).

  \begin{example}
    \label{ex:intro}
    We prove that $\lang_{A_2}(r_0) \subseteq \lang_{A_1}(q_0)$ by showing that $(r_0, q_0)$ is included in $\fsimdirect$.
    We pick $R = \{ (r_0, q_0), (r_1, q_1), (r_2, q_2) \}$
    and we observe that $(r_0, q_0) \in R \subseteq \fsimdirectF(R)$.
    Thus, by \Cref{coind}, $(r_0, q_0) \in \fsimdirect$ and by \Cref{thm:fsimdirect-sound} it follows that $\lang_{A_2}(r_0) \subseteq \lang_{A_1}(q_0)$.
  \end{example}

  \subsection{A Simple Deductive System via Parameterized Coinduction}

  Although direct simulation is a very strong notion of fairness-preserving simulation, and thus limited in applicability (we discuss these limitations in \Cref{sec:limitations-direct-sim}), its simplicity makes it a good starting point for interactive proofs.
  Although guessing a direct simulation is not difficult for small B\"uchi automata, for larger automata, and in particular automata with infinite state spaces, guessing a direct simulation might be more difficult.
  Instead, as discussed in the preliminaries, one can exploit \textit{parameterized coinduction} to obtain a set of customized reasoning rules (a deductive system) for proving that pairs of states are contained in $\fsimdirect$.

  The deductive system we propose is operating on triples $\gtriple{H}{s_1}{\preccurlyeq}{s_2}$ and $\triple{H}{s_1}{\preccurlyeq}{s_2}$ where $H$ is a relation between states representing a fragment of a direct simulation, and $s_1$ and $s_2$ are two states.
  Intuitively, a solid box $\cbox{H}$ indicates that $H$ is \emph{guarded} and cannot immediately be used, whereas a dashed box $\obox{H}$ indicates that $H$ can be used to conclude that the target state are in direct simulation.
  This notation is borrowed from \cite{for-free} and \cite{hyco}, and these triples are defined in terms of the parameterized greatest fixed point operator as follows: \begin{align*}
    \gtriple{H}{s_1}{\preccurlyeq}{s_2} &\triangleq (s_1, s_2) \in G_{\fsimdirectF}(H)\\
    \triple{H}{s_1}{\preccurlyeq}{s_2} &\triangleq (s_1, s_2) \in H \cup G_{\fsimdirectF}(H)
  \end{align*}

  Note that here, $\fsimdirectF$ denotes the monotone functor underlying the definition of $\fsimdirect$.
  Throughout the paper, we reuse this notation convention heavily: for any coinductive definition of the form $\texttt{rel} \triangleq \nu X . F(X)$, we note $\texttt{relF}$ for $F$.

  Since our triples are defined in terms of the parameterized greatest fixed point operator (see \Cref{sec:paco}), they inherit the reasoning principles of parameterized coinduction.
  In particular, $\gtriple{\emptyset}{s_1}{\preccurlyeq}{s_2} \iff (s_1, s_2) \in \fsimdirect$ and, additionally, the rules presented in \Cref{fig:direct-rules} can immediately be derived.

  \begin{figure}
    \begin{mathpar}
      \inferrule[Final]{s_2 \in \mathcal{F}_2 \\ \forall e. \forall \iostep{s_1}{e}{s'_1}. \exists \iostep{s_2}{e}{s'_2}. \ \triple{H}{s'_1}{\preccurlyeq}{s'_2}}{\gtriple{H}{s_1}{\preccurlyeq}{s_2}}\quad
      \inferrule[Step]{s_1 \notin \mathcal{F}_1\\ \forall e. \forall \iostep{s_1}{e}{s'_1}. \exists \iostep{s_2}{e}{s'_2}. \ \triple{H}{s'_1}{\preccurlyeq}{s'_2}}{\gtriple{H}{s_1}{\preccurlyeq}{s_2}}\\
      \inferrule[Cycle]{(s_1, s_2) \in H}{\triple{H}{s_1}{\preccurlyeq}{s_2}} \qquad
      \inferrule[Guard]{\gtriple{H}{s_1}{\preccurlyeq}{s_2}}{\triple{H}{s_1}{\preccurlyeq}{s_2}} \qquad
      \inferrule[Invariant]{(s_1, s_2) \in H'\\\forall (s'_1, s'_2) \in H'. \ \gtriple{H \cup H'}{s'_1}{\preccurlyeq}{s'_2}}{\gtriple{H}{s_1}{\preccurlyeq}{s_2}}
    \end{mathpar}
    \caption{Rules for $\fsimdirect$}
    \label{fig:direct-rules}
    \Description{Rules for $\fsimdirect$}
  \end{figure}

  The \textsc{Final} and \textsc{Step} rules are stepping through the two targeted automata, ensuring that every step in the left-hand automaton can be reproduced by at least one step in the right-hand automaton.
  We note that following the intuition of direct simulation, progress can only be made via \textsc{Step} and \textsc{Final} if the current left-state is rejecting ($s_1 \notin \mathcal{F}_1$) or if the current right-state is accepting ($s_2 \in \mathcal{F}_2$).
  Additionally, we note that making progress via \textsc{Step} or \textsc{Final} always releases the guard.
  The remaining rules \textsc{Cycle}, \textsc{Guard} and \textsc{Invariant} are for handling the parameter $H$.
  \textsc{Cycle} allows to conclude a proof whenever $H$ is unguarded
  and already contains the targeted pair of states.
  The rule $\textsc{Guard}$ restores the guard around an unguarded hypothesis. Finally, the rule \textsc{Invariant} extends the
  current hypothesis $H$ with a relation $H'$.
  Important, $H'$ needs to contain at least the currently targeted pair of states. Further, after applying $\textsc{Invariant}$, the proof resumes from an arbitrary pair of states in $H'$ and the guard is maintained (thus requiring to later make progress via \textsc{Final} or \textsc{Step}).

  \begin{example}
    \label{ex:intro_deductive}
    We prove that $\lang_{A_2}(r_0) \subseteq \lang_{A_1}(q_0)$ by showing that $\gtriple{\emptyset}{r_0}{\preccurlyeq}{q_0}$.
    \begin{align*}
      &\gtriple{\emptyset}{r_0}{\preccurlyeq}{q_0}\\
      \textrm{\color{gray} By \textsc{Step} and \textsc{Guard}} \impliedby&\gtriple{\emptyset}{r_1}{\preccurlyeq}{q_1} \wedge \gtriple{\emptyset}{r_2}{\preccurlyeq}{q_2}\\
      \textrm{\color{gray} By \textsc{Invariant}} \impliedby&\gtriple{(r_1, q_1)}{r_1}{\preccurlyeq}{q_1} \wedge \gtriple{ (r_2, q_2) }{r_2}{\preccurlyeq}{q_2}\\
      \textrm{\color{gray} By \textsc{Final} as $q_1,q_2 \in \mathcal{F}_2$} \impliedby&\triple{(r_1, q_1)}{r_1}{\preccurlyeq}{q_1} \wedge \triple{ (r_2, q_2) }{r_2}{\preccurlyeq}{q_2}\\
      \textrm{\color{gray} By \textsc{Cycle}}\impliedby&(r_1, q_1) \in \{(r_1, q_1)\}\wedge (r_2, q_2) \in \{ (r_2, q_2) \}
    \end{align*}
    We note that in this proof, we omitted the brackets around sets of pairs for readability (e.g., we noted $\cbox{(r_1, q_1)}$ instead of $\cbox{\{(r_1, q_1)\}}$). We adopt this notation for the remaining of the paper.
  \end{example}

  \subsection{Limitations of Direct Simulation}
  \label{sec:limitations-direct-sim}

  In the previous section, we have demonstrated how direct simulation, a simple extension of the standard notion of simulation, can be used to prove language inclusion of B\"uchi automata. Furthermore, we showed that combined with parameterized coinduction, direct simulation can be presented as a simple and intuitive deductive system to prove language containment between two B\"uchi automata.
  Nonetheless, direct simulation is still way too strong a simulation relation and it is easy to see that $\fsimdirect$ is incomplete.

  \begin{example}[Incompleteness of $\fsimdirect$]
    Let us consider the two following B\"uchi automata $A_1$ and $A_2$.
    
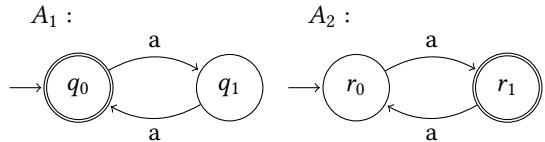
\begin{center}
      \begin{tikzpicture}[shorten >=1pt,node distance=2cm,on grid,auto,initial text=]
            \node[state,initial,accepting] (q_0) {$q_0$};
            \node[above=of q_0, xshift=-10pt, yshift=-30pt] (A1) {$A_1:$};
            \node[state] (q_1) [right=of q_0] {$q_1$};
            \path[->]
              (q_0) edge [bend left] node {a} (q_1)
              (q_1) edge [bend left] node {a} (q_0);
        \end{tikzpicture}
        \begin{tikzpicture}[shorten >=1pt,node distance=2cm,on grid,auto,initial text=]
          \node[state,initial] (q_0) {$r_0$};
          \node[above=of q_0, xshift=-10pt, yshift=-30pt] (A1) {$A_2:$};
          \node[state,accepting] (q_1) [right=of q_0] {$r_1$};
          \path[->]
            (q_0) edge [bend left] node {a} (q_1)
            (q_1) edge [bend left] node {a} (q_0);
        \end{tikzpicture}
      \end{center}

    Clearly, $\mathcal{L}(A_1) \subseteq \mathcal{L}(A_2)$ as both automata accepts exactly the language $L = \{ a^\omega \}$. However, $(q_0, r_0) \notin \fsimdirect$.
    The only accepting run for $a^\omega$ in $A_1$ visits the accepting state at even positions, whereas the only accepting run in $A_2$ only visits an accepting state at odd positions.
    To be more formal, for these two specific automata, the largest direct simulation is $ \fsimdirect = \{ (q_0, r_1), (q_1, r_0), (q_1, r_1) \}$ and $(q_0, r_0) \notin \fsimdirect$.
  \end{example}
  
  Beyond this specific example, simulation techniques are notoriously incomplete without making further assumptions on the transition systems, even without having to handle fairness conditions \cite{refinement-mappings, simple-prophecies, Baier2008PrinciplesOM}.
  In fact, any refinement of standard simulation (i.e., any relation $R \subseteq \texttt{sim}$) is necessarily incomplete for language inclusion (because \texttt{sim} is incomplete for trace inclusion, and trace inclusion is equivalent to language inclusion when marking all states as accepting).
  In spite of this inherent source of incompleteness, the simplicity of direct simulation makes it an appealing proof technique.
  Further, as demonstrated in \Cref{ex:intro}, direct simulation is well-suited for interactive deductive reasoning.

  \section{Delayed Notions of Simulation}

  \subsection{Delay Simulation}
  \label{sec:delaysim}

  Although any refinement of standard simulation is necessarily incomplete for language inclusion, direct simulation is unnecessarily strong:
  by definition, it is sensitive to the exact number of computation steps that separate visits to accepting states.
  This renders the technique essentially unusable for any verification task.
  Indeed, in practice, a program (combined with fairness assumptions)
  will typically execute for a data-dependent number of steps in between intermediate goals.
  On the other hand, specification automata are fixed and require more precisely timed visits to accepting states.
  Consequently, $\fsimdirect$ simulation is too strong to reason about programs and specifications.

  A weaker notion of fairness-preserving simulation is delay simulation \cite{delay-simulation}.
  As for direct simulation, delay simulation extends standard simulation by further enforcing that infinitely many visits to a left-accepting state are simulated by infinitely many visits to a right-accepting state. However, contrary to direct simulation, delay simulation does not force visits to accepting states to be exactly synchronized. Instead, whenever a left-accepting state is encountered, delay simulation permits postponing (for a bounded number of steps) the moment when a corresponding right-accepting state is going to be reached.
  This intuition is achieved by nesting a coinductive and an inductive relation.

  \begin{definition}[Delay Simulation]
    \label{def:delaysim}
    Let $(\mathcal{S}_1, \mathcal{E}, \mathcal{I}_1, \to_1, \mathcal{F}_1)$ and $(\mathcal{S}_2, \mathcal{E}, \mathcal{I}_2, \to_2, \mathcal{F}_2)$ be two B\"uchi automata over the same set of events $\mathcal{E}$. We define the two relations $\fsimdelayR$ and $\fsimdelayL$ as follows: \begin{align*}
      \fsimdelayR(X) &\triangleq \mu Y .\\
      \textsc{\color{gray}(R-Final)} & \quad \{ \ (s_1, s_2) \mid s_2 \in \mathcal{F}_2 \wedge \forall e. \forall \iostep{s_1}{e}{s'_1}. \exists \iostep{s_2}{e}{s'_2}. \ (s'_1, s'_2) \in X \ \} \ \cup\\
      \textsc{\color{gray}(R-Delay)} & \quad \{ \ (s_1, s_2) \mid \forall e. \forall \iostep{s_1}{e}{s'_1}. \exists \iostep{s_2}{e}{s'_2}. \ (s'_1, s'_2) \in Y \ \} \\
      \fsimdelayL &\triangleq \nu X .\\
      \textsc{\color{gray}(L-to-R)} & \quad \{ \ (s_1, s_2) \mid (s_1, s_2) \in \fsimdelayR(X) \ \} \ \cup\\
      \textsc{\color{gray}(L-Step)} & \quad \{ \ (s_1, s_2) \mid s_1 \notin \mathcal{F}_1 \wedge \forall e. \forall \iostep{s_1}{e}{s'_1}. \exists \iostep{s_2}{e}{s'_2}. \ (s'_1, s'_2) \in X \ \}
    \end{align*}
  \end{definition}

  Our formulation of delay simulation combines an inductive relation 
  $\fsimdelayR$ and a coinductive relation $\fsimdelayL$.
  The coinductive relation $\fsimdelayL$ is the entry point, and it tracks left-accepting states.
  Whenever a left-accepting state is encountered, $\fsimdelayL$
  transfers the control to $\fsimdelayR$, which intuitively requires to reach a right-accepting state.
  More precisely, $\fsimdelayR$ contains all pairs of states from which, no matter the steps taken by the left-hand automaton, the right-hand automaton can always mimic these steps in such a way that a right-accepting state is eventually found.
  Once such a state is attained, $\fsimdelayR$ transfers the control back to $\fsimdelayL$.
  It is important to note that when switching from $\fsimdelayL$ to $\fsimdelayR$ (see {\color{gray}(\textsc{L-to-R})} in \Cref{def:delaysim}), no computation steps need to be taken.
  However, switching back from $\fsimdelayR$ to $\fsimdelayL$ (see {\color{gray}(\textsc{R-final})} in \Cref{def:delaysim}) does require taking a step. Without this subtle difference, $\fsimdelayL$
  would contain any pair of states $(s_1, s_2)$ with $s_2 \in \mathcal{F}_2$, and it would not be a sound proof technique for language inclusion.
  In general, defining fair simulation relations require extra care.
  Errors like off-by-one-step errors are easily introduced and one has to be particularly careful to make sure that \emph{progress} is always made before corecursing.

  We further note that there are other equivalent ways to define $\fsimdelay$.
  For example, instead of requiring to take a step before switching back to $\fsimdelayL$ in {\color{gray}(\textsc{R-Final})}, one could also take a step before switching to $\fsimdelayR$ in {\color{gray}(\textsc{L-to-R})} and add a third rule allowing to
  corecurse when a right-accepting state is encountered in $\fsimdelayL$.
  In principle, the style of definition does not matter much.
  However, we observed that in practice, the exact choice of formalization can significantly impact the complexity of the soundness proof.

  Before proving the soundness of $\fsimdelay$ for language inclusion, we make two useful observations. First, we observe that $\fsimdelayR(\fsimdelayL) \subseteq \fsimdelayL$.

  \begin{lemma}
     \label{lem:delaysim_reach}
     $\fsimdelayR(\fsimdelayL) \subseteq \fsimdelayL$ 
   \end{lemma}
   \begin{proof}
    Let $(s_1, s_2) \in \fsimdelayR(\fsimdelayL)$,
    we have to show that $(s_1, s_2) \in \fsimdelayL$.
    By unfolding $\fsimdelayL$, it is equivalent to prove $(s_1, s_2) \in \fsimdelayFL(\fsimdelayL)$.
    Using the {\color{gray}(\textsc{L-to-R})} disjunct, it is enough to prove
    $(s_1, s_2) \in \fsimdelayR(\fsimdelayL)$.
    This is exactly our assumption.
  \end{proof}

  Another useful observation is that $\fsimdelayL$ is a postfixed point of $\texttt{simF}$. By \Cref{coind}, this implies that $\fsimdelayL$ is stronger than standard simulation (i.e., $\fsimdelayL \subseteq \texttt{sim}$).

  \begin{lemma}
    \label{lem:delaysim_step}
    $(s_1, s_2) \in \fsimdelayL \implies \forall e. \forall \iostep{s_1}{e}{s'_1}. \exists \iostep{s_2}{e}{s'_2}. \ (s'_1, s'_2) \in \fsimdelayL$
  \end{lemma}
  \begin{proof}
    Suppose $(s_1, s_2) \in \fsimdelayL$.
    By unfolding the greatest fixed point we know that in particular $(s_1, s_2) \in \fsimdelayFL(\fsimdelayL)$.
    We then proceed by case analysis on $\fsimdelayFL$.
    The {\color{gray}(\textsc{L-Step})} case is trivial.
    In the {\color{gray}(\textsc{L-to-R})} case, we have $(s_1, s_2) \in \fsimdelayR(\fsimdelayL)$ and by 
    \Cref{lem:delaysim_reach} it immediately follows that $(s_1, s_2) \in \fsimdelayL$.
  \end{proof}

  Using \Cref{lem:delaysim_reach} and \Cref{lem:delaysim_step}, we are ready to prove the soundness of $\fsimdelay$.

  \begin{theorem}[Soundness of $\fsimdelay$]
    \label{thm:fsimdelay-sound}
    $(s_1, s_2) \in \fsimdelayL \implies \lang(s_1) \subseteq \lang(s_2)$
  \end{theorem}
  \begin{proof} (Sketch)
    The proof proceeds by coinduction on $\lang(s_2)$, and then by induction on the inductive part of $\lang(s_1)$.
    In the cases where the state $s_1$ is final, we do a second induction on $\fsimdelayR$, applying \Cref{lem:delaysim_step} along the way until a right-accepting state is reached. We then conclude by coinduction.
    In the case where $s_1$ is not final, we conclude using \Cref{lem:delaysim_step} and the induction hypothesis.
  \end{proof}

  \subsection{Basic Deductive System and Examples}

  Instead of giving examples of proofs by delay simulation using the raw definition of $\fsimdelay$, we instead directly present its associated deductive system.
  The deductive system is operating on the following three kinds of triples: \begin{align*}
    \gtripleDelay{H}{s_1}{\relL}{s_2} &\triangleq (s_1, s_2) \in G_{\fsimdelayFL}(H)\\
    \tripleDelay{H}{s_1}{\relL}{s_2} &\triangleq (s_1, s_2) \in H \cup G_{\fsimdelayFL}(H)\\
    \gtripleDelay{H}{s_1}{\relR}{s_2} &\triangleq (s_1, s_2) \in \fsimdelayR(H \cup G_{\fsimdelayFL}(H))
  \end{align*}

  Here, the subscripts \texttt{L} and \texttt{R} indicates whether we are currently
  tracking \textbf{\texttt{L}}eft-accepting states (with $\fsimdelayL$) or
  currently searching for a \textbf{\texttt{R}}ight-accepting one (with $\fsimdelayR)$. As for $\fsimdirect$, specializing the rules of parameterized coinduction (see \Cref{Paco}) immediately gives a basic set of reasoning principles which we list in
  \Cref{fig:rules_esim}.

  \begin{figure}[ht!]
    \begin{mathpar}
      \inferrule[L-to-R]{\gtripleDelay{H}{s_1}{\relR}{s_2}}{\gtripleDelay{H}{s_1}{\relL}{s_2}}\qquad
      \inferrule[L-Step]{s_1 \notin \mathcal{F}_1\\\forall e. \forall \iostep{s_1}{e}{s'_1}. \exists \iostep{s_2}{e}{s'_2}. \ \tripleDelay{H}{s'_1}{\relL}{s'_2}}{\gtripleDelay{H}{s_1}{\relL}{s_2}}\\
      \inferrule[R-Delay]{\forall e. \forall \iostep{s_1}{e}{s'_1}. \exists \iostep{s_2}{e}{s'_2}. \ \gtripleDelay{H}{s'_1}{\relR}{s'_2}}{\gtripleDelay{H}{s_1}{\relR}{s_2}}\qquad
      \inferrule[R-Final]{s_2 \in \mathcal{F}_2\\\forall e. \forall \iostep{s_1}{e}{s'_1}. \exists \iostep{s_2}{e}{s'_2}. \ \tripleDelay{H}{s'_1}{\relL}{s'_2}}{\gtripleDelay{H}{s_1}{\relR}{s_2}}\\
      \inferrule[L-Cycle]{(s_1, s_2) \in H}{\tripleDelay{H}{s_1}{\relL}{s_2}} \qquad
      \inferrule[L-Guard]{\gtripleDelay{H}{s_1}{\relL}{s_2}}{\tripleDelay{H}{s_1}{\relL}{s_2}} \qquad
      \inferrule[L-Invariant]{(s_1, s_2) \in H'\\\forall (s'_1, s'_2) \in H'. \ \gtripleDelay{H \cup H'}{s'_1}{\relL}{s'_2}}{\gtripleDelay{H}{s_1}{\relL}{s_2}}
    \end{mathpar}
    \caption{Basic rules for $\fsimdelay$}
    \label{fig:rules_esim}
    \Description{Rules}
  \end{figure}

  The rules \textsc{L-to-R}, \textsc{L-Step}, \textsc{R-Delay} and \textsc{R-Final} are immediately obtained by unfolding the definition of $\fsimdelayR$ and by instantiating the \textsc{Step} rule of parameterized coinduction with $\fsimdelayFL$ as the underlying monotone functor.
  We note that without further assumptions, the context $H$ can only be exploited and modified (via the rules \textsc{L-Cycle}, \textsc{L-Guard}, and \textsc{L-Invariant}) when the goal is an \textsc{L}-triple $\gtripleDelay{H}{s_1}{\relL}{s_2}$ or $\tripleDelay{H}{s_1}{\relL}{s_2}$!
  Indeed, when the goal is a \texttt{R}-triple $\gtripleDelay{H}{s_1}{\relR}{s_2}$, the context $H$ is guarded by an application of the inductive relation $\fsimdelayR$, thus forbidding to exploit the rules of parameterized coinduction.
  In \Cref{sec:more-rules}, we discuss this limitation, and we show that with a little more work, \texttt{R}-triples can also be equipped with additional rules to modify the context. For now, we focus on an example using only the basic rules of \Cref{fig:rules_esim}.
  
  \begin{example}
  Recall the following two automata from the previous section:
  
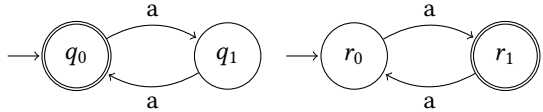
\begin{center}
    \begin{tikzpicture}[shorten >=1pt,node distance=2cm,on grid,auto,initial text=]
      \node[state,initial,accepting] (q_0) {$q_0$};
      \node[state] (q_1) [right=of q_0] {$q_1$};
      \path[->] (q_0) edge [bend left] node {a} (q_1)
        (q_1) edge [bend left] node {a} (q_0);
      \end{tikzpicture}
    \begin{tikzpicture}[shorten >=1pt,node distance=2cm,on grid,auto,initial text=]
      \node[state,initial] (q_0) {$r_0$};
      \node[state,accepting] (q_1) [right=of q_0] {$r_1$};
      \path[->] (q_0) edge [bend left] node {a} (q_1)
        (q_1) edge [bend left] node {a} (q_0);
      \end{tikzpicture}
  \end{center}
  
  We already discussed that there exists no direct simulation between these two automata. However, they can be proven to be language included by delay simulation.
  In particular, we prove that $\gtripleDelay{\emptyset}{q_0}{\relL}{r_0}$.

  We start by using the \textsc{L-Invariant} rule to add the current states to the context.
  \begin{align*}
    &\gtripleDelay{\emptyset}{q_0}{\relL}{r_0} \\
    \textrm{\color{gray} By \textsc{L-Invariant}} \impliedby&\gtripleDelay{(q_0,r_0)}{q_0}{\relL}{r_0} 
  \end{align*}
  As $s_0$ is final, the only rule we can apply to make further progress is the \textsc{L-to-R} rule, thus replacing the current L-triple with an R-triple, and initiating a search for a right-accepting state.
  \begin{align*}
    \textrm{\color{gray} By \textsc{L-to-R}} \impliedby&\gtripleDelay{(q_0,r_0)}{q_0}{\relR}{r_0} 
  \end{align*}
  As $r_0$ is not final, the only option is to apply rule \textsc{R-Delay} to investigate all possible ways to transition out of $(q_0, r_0)$. Here, the only option is to move to $(q_1, r_1)$ via the event $a$.
  \begin{align*}
    \textrm{\color{gray} By \textsc{R-Delay}} \impliedby&\gtripleDelay{(q_0,r_0)}{q_1}{\relR}{r_1} 
  \end{align*}
  Since $r_1$ is final, we can circle back to an L-triple with \textsc{R-Final}.
  \begin{align*}
    \textrm{\color{gray} By \textsc{R-Final}} \impliedby&\tripleDelay{ (q_0, r_0) }{q_0}{\relL}{r_0} 
  \end{align*}
  Now, we are back in $(q_0 ,r_0)$ and importantly, the guard around the context is released.
  Hence, we can use the \textsc{L-Cycle} rule to conclude the proof.
  \begin{align*}
    \textrm{\color{gray} By \textsc{L-Cycle}} \impliedby&(q_0,r_0) \in \{(q_0,r_0)\} 
  \end{align*}
\end{example}

  \subsection{Additional Reasoning Rules}
  \label{sec:more-rules}

  Strictly by following the definition of our triples, the rules of parameterized coinduction cannot be applied to \texttt{R}-triples.
  In particular, with the current definition of $\fsimdelay$ it would be incorrect to use the context $H$ to prove a \texttt{R}-triple.
  Indeed, this would allow us to exploit any cycle
  to abandon an obligation to reach a right-accepting state. Even though using the context during a proof of an \texttt{R}-triple is unsound, extending it by accumulating pairs of states visited during the proof of an \texttt{R}-triple can be a convenient feature to avoid redundant proof steps. We demonstrate this observation with an example.

  \pagebreak
  \begin{example}
    Consider the following two B\"uchi automata:
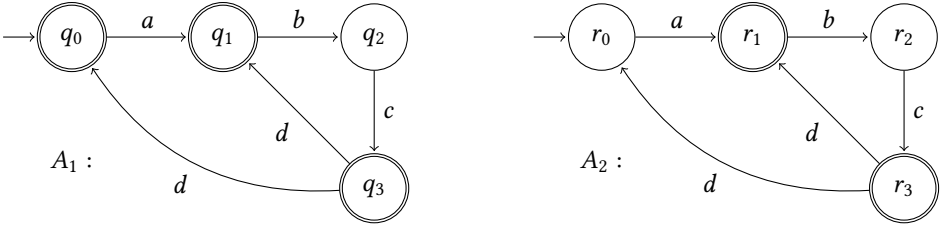
    \begin{center}
      \begin{tikzpicture}[shorten >=1pt,node distance=2cm,on grid,auto,initial text=]
        \node[initial, state, accepting] (A) {$q_0$};
        \node[below=of A, yshift=10pt] (name) {$A_1:$};
        \node[state, right=of A, accepting] (B) {$q_1$};
        \node[state, right=of B] (C) {$q_2$};
        \node[state, below=of C, accepting] (D) {$q_3$};
        \path[->]
          (A) edge node{$a$} (B)
          (B) edge node{$b$} (C)
          (C) edge node{$c$} (D)
          (D) edge node{$d$} (B)
          (D) edge[bend left] node{$d$} (A);
      \end{tikzpicture}%
      \qquad\qquad
      \begin{tikzpicture}[shorten >=1pt,node distance=2cm,on grid,auto,initial text=]
        \node[initial, state] (A) {$r_0$};
        \node[below=of A, yshift=10pt] (name) {$A_2:$};
        \node[state, right=of A, accepting] (B) {$r_1$};
        \node[state, right=of B] (C) {$r_2$};
        \node[state, below=of C, accepting] (D) {$r_3$};
        \path[->]
          (A) edge node{$a$} (B)
          (B) edge node{$b$} (C)
          (C) edge node{$c$} (D)
          (D) edge node{$d$} (B)
          (D) edge[bend left] node{$d$} (A);
      \end{tikzpicture}
    \end{center}
    Both automata have exactly the same states and transitions. However, $A_1$ has its initial state marked as accepting but not $A_2$. This difference does not change the language recognized by $A_2$ and $\lang(A_1) \subseteq \lang(A_2)$ (in fact, we even have $\lang(A_1) = \lang(A_2)$). We prove that $\lang(A_1) \subseteq \lang(A_2)$ by delay simulation. \begin{align*}
      &\gtripleDelay{\emptyset}{q_0}{\relL}{r_0}\\
      \textrm{\color{gray}By \textsc{L-Invariant}}
      \impliedby&\gtripleDelay{(q_0, r_0)}{q_0}{\relL}{r_0}\\
      \textrm{\color{gray}By \textsc{L-to-R} and \textsc{R-Delay}}
      \impliedby&\gtripleDelay{(q_0, r_0)}{q_1}{\relR}{r_1}\\
      \textrm{\color{gray}By 2 times \textsc{R-Delay}}
      \impliedby&\gtripleDelay{(q_0, r_0)}{q_3}{\relR}{r_3}\\
      \textrm{\color{gray}By \textsc{R-Final}}
      \impliedby&\tripleDelay{(q_0, r_0)}{q_1}{\relL}{r_1} \wedge \tripleDelay{(q_0, r_0)}{q_0}{\relL}{r_0}
    \end{align*}
    At this point, the goal $\tripleDelay{(q_0, r_0)}{q_0}{\relL}{r_0}$ can be immediately discharged by $\textsc{L-Cycle}$.
    However, we can not yet conclude for $\tripleDelay{(q_0, r_0)}{q_1}{\relL}{r_1}$. We therefore have to extend our current invariant. \begin{align*}
      &\tripleDelay{(q_0, r_0)}{q_1}{\relL}{r_1}\\
      \textrm{\color{gray}By \textsc{L-Guard} and \textsc{L-Invariant}}
      \impliedby&\gtripleDelay{(q_0, r_0), (q_1, r_1)}{q_1}{\relL}{r_1}\\
      \textrm{\color{gray}By \textsc{L-to-R} and 2 times \textsc{R-Delay}}
      \impliedby&\gtripleDelay{(q_0, r_0), (q_1, r_1)}{q_3}{\relL}{r_3}\\
      \textrm{\color{gray}By \textsc{R-Final} and \textsc{L-Cycle}}
      \impliedby&(q_0, r_0) \in \{ (q_0, r_0), (q_1, r_1)\} \wedge (q_1, r_1) \in \{ (q_0, r_0), (q_1, r_1)\}
    \end{align*}
    Observe that the last steps of reasoning are duplicated! Indeed, we already encountered
    the pair of states $(q_1, r_1)$ earlier in the proof when we proved the \texttt{R}-triple $\gtripleDelay{ (q_0, r_0)}{q_1}{\relR}{r_1}$. Intuitively, at this point, we would have wanted to extend our context with $(q_1, r_1)$, thus allowing us to immediately discharge the two goals generated by the first application of \textsc{R-Final}.
    Unfortunately, the current rules do not allow to extend the context while establishing a \texttt{R}-triple.
  \end{example}

  Even though the rules of parameterized coinduction do not immediately allow us to accumulate pairs of states visited during the proof of a \texttt{R}-triples, by definition all such pairs are still guaranteed to reach a right-accepting state. Intuitively, nothing should prevent us from remembering this fact by accumulating pairs of states in the context.
  In the following, we prove that this intuition is indeed correct. More precisely, we prove that the following \textsc{R-Invariant} rule is sound: \begin{mathpar}
    \inferrule[R-Invariant]{(s_1, s_2) \in H'\\ H' \subseteq \mathcal{F}_1 \times \mathcal{S}_2\\\forall (s'_1, s'_2) \in H'. \ \gtripleDelay{H \cup H'}{s_1}{\relR}{s_2}}{\gtripleDelay{H}{s_1}{\relR}{s_2}}
  \end{mathpar}

  To establish the soundness of $\textsc{R-Invariant}$, we first observe that when $s_1 \in \mathcal{F}_1$, $\gtriple{H}{s_1}{\relL}{s_2}$ also implies $\gtriple{H}{s_1}{\relR}{s_2}$.

  \begin{lemma}
    \label{lem:rinvariant-aux}
    $s_1 \in \mathcal{F}_1 \implies \gtriple{H}{s_1}{\relL}{s_2} \implies \gtriple{H}{s_1}{\relR}{s_2}$
  \end{lemma}\begin{proof}
    Suppose $s_1 \in \mathcal{F}_1$ and $\gtriple{H}{s_1}{\relR}{s_2}$.
    By unfolding the definition of \texttt{L}-triples, we have two cases.
    Either $(s_1, s_2) \in \fsimdelayR(H \cup G_{\fsimdelayFL}(H))$. This is exactly the definition of $\gtriple{H}{s_1}{\relR}{s_2}$.
    Otherwise, $s_1 \notin \mathcal{F}_1$ and for every successor $s'_1$ of $s_1$, there is a successor of $s'_2$ of $s_2$ with $\triple{H}{s_1}{\relL}{s_2}$. This case is contradictory with the assumption that $s_1 \in \mathcal{F}_1$.
  \end{proof}

  \begin{lemma}
    The rule $\textsc{R-Invariant}$ is sound.
  \end{lemma}
  \begin{proof}
    Let (i) $H' \subseteq \mathcal{F}_1 \times \mathcal{S}_2$ and suppose (ii) that $(s_1, s_2) \in H'$. Further, assume that (iii) $\forall (s'_1, s'_2) \in H'. \ \gtripleDelay{H \cup H'}{s'_1}{\relR}{s'_2}$.
    From these assumptions, we have to derive $\gtripleDelay{H}{s_1}{\relR}{s_2}$.
    By (i) and (ii) we can deduce that $s_1 \in \mathcal{F}_1$. We can then use \Cref{lem:rinvariant-aux} to establish the following chain of implications: \begin{align*}
      &\gtripleDelay{H}{s_1}{\relR}{s_2}\\
      \textrm{\color{gray}By \Cref{lem:rinvariant-aux} and $s_1 \in \mathcal{F}_1$}
      \impliedby&\gtripleDelay{H}{s_1}{\relL}{s_2}\\
      \textrm{\color{gray}By \textsc{L-Invariant} using $H'$}
      \impliedby&\forall (s'_1, s'_2) \in H'. \ \gtripleDelay{H \cup H'}{s'_1}{\relL}{s'_2}\\
      \textrm{\color{gray}By \textsc{L-to-R}}
      \impliedby&\forall (s'_1, s'_2) \in H'. \ \gtripleDelay{H \cup H'}{s'_1}{\relR}{s'_2}
    \end{align*}
    The last statement is exactly (iii), which concludes the proof.
  \end{proof}

  Using the new rule \textsc{R-Invariant}, we can now rework our previous example by accumulating $(q_1, r_1)$ into the context the first time it is encountered. The corresponding proof is summarized as follows: \begin{align*}
    &\gtripleDelay{\emptyset}{q_0}{\relL}{r_0}\\
    \textrm{\color{gray}By \textsc{L-Invariant}}
    \impliedby&\gtripleDelay{(q_0, r_0)}{q_0}{\relL}{r_0}\\
    \textrm{\color{gray}By \textsc{L-to-R} and \textsc{R-Delay}}
    \impliedby&\gtripleDelay{(q_0, r_0)}{q_1}{\relR}{r_1}\\
    \textrm{\color{gray}\textbf{By \textsc{R-Invariant}} and $q_1 \in \mathcal{F}_1$}
    \impliedby&\gtripleDelay{(q_0, r_0), (q_1, r_1)}{q_1}{\relR}{r_1}\\
    \textrm{\color{gray}By 2 times \textsc{R-Delay}}
    \impliedby&\gtripleDelay{(q_0, r_0), (q_1, r_1)}{q_3}{\relR}{r_3}\\
    \textrm{\color{gray}By \textsc{R-Final}}
    \impliedby&\tripleDelay{(q_0, r_0), (q_1, r_1)}{q_1}{\relL}{r_1} \wedge \tripleDelay{(q_0, r_0), (q_1, r_1)}{q_0}{\relL}{r_0}\\
    \textrm{\color{gray}By \textsc{L-Cycle}}
    \impliedby&(q_0, r_0) \in \{ (q_0, r_0), (q_1, r_1) \} \wedge (q_1, r_1) \in \{ (q_0, r_0), (q_1, r_1) \}
  \end{align*}

  \subsection{Right-Biased Delay Simulation}

  In the definition of delay simulation, it is very important that the 
  disjunct \textsc{L-Step} is \emph{not} in the scope of the inner least fixed point.
  Otherwise, the relation does not provide a sound technique for language inclusion. Consider the following (incorrect) reformulation of delay simulation where the inner least fixed point operator has been pulled back to be at the same level as the outer greatest fixed point.{\color{red}
  \begin{align*}
    \texttt{wrong} &\triangleq \nu X . \mu Y .\\
      \textsc{(Final)} & \quad \{ \ (s_1, s_2) \mid s_2 \in \mathcal{F}_2 \wedge \forall e. \forall \iostep{s_1}{e}{s'_1}. \exists \iostep{s_2}{e}{s'_2}. \ (s'_1, s'_2) \in X \ \} \ \cup\\
      \textsc{(Delay)} & \quad \{ \ (s_1, s_2) \mid \forall e. \forall \iostep{s_1}{e}{s'_1}. \exists \iostep{s_2}{e}{s'_2}. \ (s'_1, s'_2) \in Y \ \} \ \cup\\
      \textsc{(Step)} & \quad \{ \ (s_1, s_2) \mid s_1 \notin \mathcal{F}_1 \wedge \forall e. \forall \iostep{s_1}{e}{s'_1}. \exists \iostep{s_2}{e}{s'_2}. (s_1, s_2) \in X \ \}
  \end{align*}}

  This definition essentially merges $\fsimdelayL$ and $\fsimdelayR$ into a single coinductive-inductive predicate. Unfortunately, it is not too difficult to see that this attempt at "simplifying" $\fsimdelay$ gives a notion of simulation that is unsound for language inclusion. As a counterexample, consider the following two automata: 
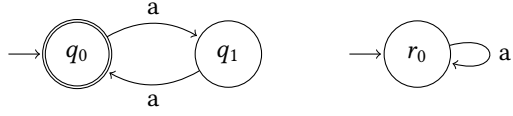
\begin{center}
    \begin{tikzpicture}[shorten >=1pt,node distance=2cm,on grid,auto,initial text=]
      \node[state,initial,accepting] (q_0) {$q_0$};
      \node[state] (q_1) [right=of q_0] {$q_1$};
      \node[state,initial, right=of q_1, xshift=.5cm] (r_0) {$r_0$};
      \path[->]
        (q_0) edge[bend left] node {a} (q_1)
        (q_1) edge[bend left] node {a} (q_0);
      \path[->]
        (r_0) edge[loop right] node {a} (r_0);
    \end{tikzpicture}
  \end{center}
  Clearly, $\lang(q_0) = \{ a^\omega\}$ and it is not included in $\lang(r_0) = \emptyset$.
  However, we can prove that $(q_0, r_0) \in \texttt{wrong}$ using $\{ (q_0, r_0), (q_1, r_0) \}$ as a coinduction hypothesis.
  Indeed, from $(q_0, r_0)$ it suffices to use \textsc{Delay} once to reach $(q_1, r_0)$.
  From there, since $q_1$ is not final, we can use \textsc{Step} to cycle back to $(q_0, r_0)$ and conclude, even though we never used \textsc{Final} to show that an accepting state can be reached in the right-hand automaton.
  
  To fix this wrong definition, we can either 
  remove the \textsc{Delay} disjunct,
  or the \textsc{Step} disjunct.
  Removing \textsc{Delay} gives back direct simulation.
  However, if we remove \textsc{Step} and keep \textsc{Delay}, we obtain the following \emph{right-biased} notion of simulation (we note $\texttt{fsim}_\texttt{rb}$).

  \begin{definition}[Right-biased Simulation]
    \begin{align*}
      \texttt{fsim}_\texttt{rb} \triangleq \nu X. \mu Y.\\
      \textsc{\color{gray}(Final)} & \quad \{ \ (s_1, s_2) \mid s_2 \in \mathcal{F}_2 \wedge \forall e. \forall \iostep{s_1}{e}{s'_1}. \exists \iostep{s_2}{e}{s'_2}. \ (s'_1, s'_2) \in X \ \} \ \cup\\
      \textsc{\color{gray}(Delay)} & \quad \{ \ (s_1, s_2) \mid \forall e. \forall \iostep{s_1}{e}{s'_1}. \exists \iostep{s_2}{e}{s'_2}. \ (s'_1, s'_2) \in Y \ \}
    \end{align*}
  \end{definition}

  We note that by removing the disjunct {\color{gray}(\textsc{Step})}, we lose the ability to exploit the fairness assumption of the left-hand automaton (hence the name \emph{right-biased}).
  Indeed, {\color{gray}(\textsc{Step})} allowed us to use the coinduction hypothesis after taking a computation step from a non-accepting left-state. This enables to use
  cyclic reasoning to ignore non-accepting loops in the left-hand automaton.
  Without {\color{gray}(\textsc{Step})}, we have no choice but to visit a right-accepting state in order to eventually get access to the coinduction hypothesis using {\color{gray}(\textsc{Final})}.

  The relation $\texttt{fsim}_\texttt{rb}$ gives a sound technique for language inclusion.
  However, since it ignores the fairness condition of the left-hand automaton, it is more accurately phrased as a technique to guarantee that any \emph{trace} of the left-hand automaton is in the language of the right-hand one. \begin{theorem}
    $(s_1, s_2) \in \texttt{fsim}_\texttt{rb} \implies \mathit{Traces}(s_1) \subseteq \lang(s_2)$
  \end{theorem}
  \begin{proof}(Sketch)
    Clearly, $\texttt{fsim}_\texttt{rb}$ is contained in $\texttt{sim}$ and therefore, it guarantees trace inclusion.
    Further, the disjunct {\color{gray}(\textsc{Final})} enforces that any execution of the left-hand automaton is simulated by an execution of the right-hand automaton that steps through an $\mathcal{F}_2$-state infinitely many times.
    This guarantees that traces of $s_1$ are in the language of $s_2$.
  \end{proof}

  Interestingly, we observe that even without the ability to exploit fairness assumptions, $\texttt{fsim}_\texttt{rb}$ can still be used as a general method for interactive proofs of arbitrary specifications expressed as a B\"uchi automaton. In particular, this includes any specification expressed in Linear Temporal Logic \cite{Baier2008PrinciplesOM}.

  As for the previous relations presented so far, $\texttt{fsim}_\texttt{rb}$ can be presented in the form of a deductive system. The two core rules are the following \textsc{Final} and \textsc{Delay} rules:
  \begin{mathpar}
    \inferrule[Final]{s_2 \in \mathcal{F}_2\\\forall e. \forall \iostep{s_1}{e}{s'_1}.\exists \iostep{s_2}{e}{s'_2}. \ \triplerb{H}{s'_1}{\preccurlyeq}{s'_2}}{\gtriplerb{H}{s_1}{\preccurlyeq}{s_2}}\qquad
    \inferrule[Delay]{\forall e. \forall \iostep{s_1}{e}{s'_1}.\exists \iostep{s_2}{e}{s'_2}. \ \gtriplerb{H}{s'_1}{\preccurlyeq}{s'_2}}{\gtriplerb{H}{s_1}{\preccurlyeq}{s_2}}
  \end{mathpar}
  
  We note that while \textsc{Final} releases the guard around the context, \textsc{Delay} does not! Naturally, as in \Cref{fig:direct-rules} and \Cref{fig:rules_esim}, additional rules $\textsc{Cycle}$, $\textsc{Guard}$, and \textsc{Invariant} can also be derived.

  \tikzset{shorten >=1pt,node distance=2cm,on grid,auto,initial text=}

  \newcommand{\pstate}[1]{\node[state, rectangle, rounded corners, #1]}
  \newcommand{\ipstate}[1]{\node[initial, state, rectangle, rounded corners, #1]}
  \newcommand{\fpstate}[1]{\node[state, rectangle, rounded corners, accepting, #1]}
  \newcommand{\ifpstate}[1]{\node[initial, state, rectangle, rounded corners, accepting, #1]}

  \begin{example}
    As an example, consider the following program (top left), its representation as an infinite-state transition system (on the right), and its specification (bottom left):
    \begin{center}
      \begin{minipage}{.4\textwidth}
      \begin{tabular}{l}
        \underline{Program:}\\
        \\
        (1)\quad$x \gets 0$\\
        (2)\quad\textbf{loop}:\\
        (3)\quad\quad$x \gets n \in \mathbb{N}$\\
        (4)\quad\quad$\textbf{while}(x > 0): x \gets x - 1$\\
        (5)\quad\quad$\textbf{print}("done")$
      \end{tabular}%
      \vskip1\baselineskip
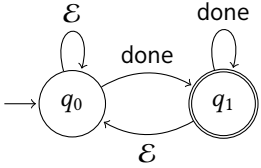
      \begin{center}
        \underline{Specification:}
        \vskip.5\baselineskip
        \begin{tikzpicture}
          \node[initial, state] (A) {$q_0$};
          \node[state, right=of A, accepting] (B) {$q_1$};
          \path[->]
            (A) edge[bend left] node{$\texttt{done}$} (B)
            (B) edge[bend left] node{$\mathcal{E}$} (A)
            (A) edge[loop above] node{$\mathcal{E}$} (A)
            (B) edge[loop above] node{$\texttt{done}$} (B);
        \end{tikzpicture}
      \end{center}
      \end{minipage}%
      \begin{minipage}{.6\textwidth}
        
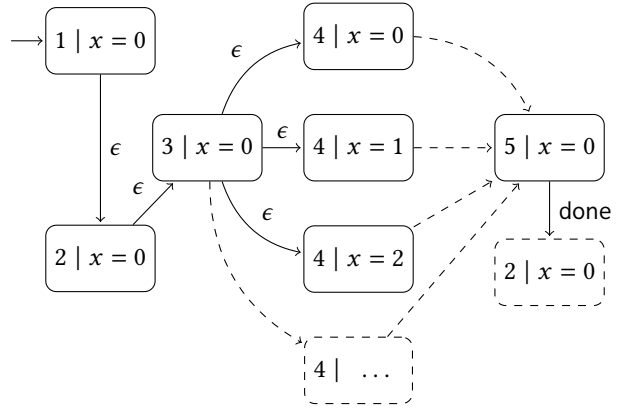
\begin{tikzpicture}
          \ipstate{} (A) {$1 \mid x = 0$};
          \pstate{below=of A, yshift=-25pt} (A') {$2 \mid x = 0$};
          \pstate{below right=of A} (B) {$3 \mid x = 0$};
          \pstate{right=of B} (D) {$4 \mid x = 1$};
          \pstate{above=of D, yshift=-15pt} (C) {$4 \mid x = 0$};
          \pstate{below=of D, yshift=15pt} (E) {$4 \mid x = 2$};
          \pstate{right=of D, xshift=15pt} (F) {$5 \mid x = 0$};
          \pstate{below=of E, yshift=15pt, dashed} (G) {$4 \mid \ \ \ldots \ \ $};
          \pstate{right=of E, dashed, xshift=15pt, yshift=-5pt} (H) {$2 \mid x = 0 $};
          \draw[->]
            (A) edge node{$\epsilon$} (A')
            (A') edge node{$\epsilon$} (B)
            (B) edge[bend left] node{$\epsilon$} (C)
            (B) edge node{$\epsilon$} (D)
            (B) edge[bend right] node{$\epsilon$} (E)
            (C) edge[dashed, bend left] (F)
            (D) edge[dashed] (F)
            (E) edge[dashed] (F)
            (B) edge[dashed, bend right] (G)
            (G) edge[dashed] (F)
            (F) edge node{$\texttt{done}$} (H);
        \end{tikzpicture}
      \end{minipage}
    \end{center}
    The program operates a single variable $x \in \mathbb{N}$, initialized to be $0$.
    Then, it enters an infinite loop. At each iteration of this outer loop, a new value $n \in \mathbb{N}$ is determined. The program then counts down to $0$ from $n$ and prints the message "done".
    This example could model, for example, a reactive system continuously receiving user inputs ($x$), and processing it before providing an answer.
    The specification we wish to prove is that the program is \textit{\enquote{done}} infinitely many times. In LTL notation, we want to prove $\LTLsquare\LTLfinally\texttt{done}$. We give a proof using the rules of $\texttt{fsim}_\texttt{rb}$. \begin{align*}
      &\gtriplerb{\emptyset}{(1, 0)}{\preccurlyeq}{q_0}\\
      \textrm{\color{gray}By \textsc{Delay} (2 times)}
      \impliedby&\gtriplerb{\emptyset}{(3, 0)}{\preccurlyeq}{q_0}
    \end{align*}
    We use the invariant $I = \{ \ ((3, 0), q_0 ) \ \}$.
    \begin{align*}
      \textrm{\color{gray}By \textsc{Invariant}}
      \impliedby&\gtriplerb{$I$}{(3, 0)}{\preccurlyeq}{q_0}\\
      \textrm{\color{gray}By \textsc{Delay} and $x \gets n \in \mathbb{N}$}
      \impliedby&\forall n \in \mathbb{N}. \ \gtriplerb{$I$}{(4, n)}{\preccurlyeq}{q_0}
    \end{align*}
    From there, the proof goes by induction on $n \in \mathbb{N}$. The case of $n = 0$ is simple: \begin{align*}
      &\gtriplerb{I}{(4, 0)}{\preccurlyeq}{q_0}\\
      \textrm{\color{gray}By \textsc{Delay} and since $\neg (0 > 0)$}
      \impliedby&\gtriplerb{I}{(5, 0)}{\preccurlyeq}{q_0}\\
      \textrm{\color{gray}By \textsc{Delay} and by printing \textit{\enquote{done}}}
      \impliedby&\gtriplerb{I}{(2, 0)}{\preccurlyeq}{q_1}\\
      \textrm{\color{gray}By \textsc{Final} (since $q_1$ is final)}
      \impliedby&\triplerb{I}{(3, 0)}{\preccurlyeq}{q_0}\\
      \textrm{\color{gray}By \textsc{Cycle}}
      \impliedby&((3, 0), q_0) \in I
    \end{align*}
    Finally, we cover the inductive case.
    By induction, we get to assume $\gtriplerb{I}{(4, n)}{\preccurlyeq}{q_0}$, and we have to show $\gtriplerb{I}{(4, n + 1)}{\preccurlyeq}{q_0}$.
    \begin{align*}
      &\gtriplerb{I}{(4, n + 1)}{\preccurlyeq}{q_0}\\
      \textrm{\color{gray}By \textsc{Delay}, $n + 1 > 0$ and $x \gets x - 1$}
      \impliedby&\gtriplerb{I}{(4, n)}{\preccurlyeq}{q_0}
    \end{align*}
    We can now conclude using our induction hypothesis.
  \end{example}

  \section{Simulations with Repeated Delay}

  \subsection{The Problem of Spurious Left-Accepting States}

  In the previous section, we observed that for transition systems without fairness assumptions, a very simple notion of right-biased simulation is sufficient to prove liveness properties encoded as B\"uchi automata.
  For more challenging examples, we might however need to exploit fairness assumptions in order to prove that a liveness objective is accomplished. Direct simulation and delay simulation both have the ability to exploit fairness assumptions, but in a very restricted way.
  The main limitation of delay simulation is that \emph{as soon} as an accepting state is visited in the left, it forces to visit an accepting state in the right at a later point. In particular, this means that a left automaton with unnecessarily many accepting states will be more difficult to deal with. As an extreme example of this phenomenon, let us consider the following example:

  \begin{example}
    Consider the following two automata.
  
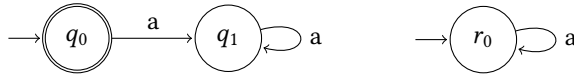
\begin{center}
      \begin{tikzpicture}[shorten >=1pt,node distance=2cm,on grid,auto,initial text=]
        \node[state,initial,accepting]  (q_0)                      {$q_0$};
        \node[state]          (q_1) [right=of q_0] {$q_1$};
        \path[->] (q_0) edge              node        {a} (q_1)
                  (q_1) edge [loop right] node        {a} ();
      \end{tikzpicture}
      $\quad \quad$
      \begin{tikzpicture}[shorten >=1pt,node distance=2cm,on grid,auto,initial text=]
              \node[state,initial] (q_0) {$r_0$};
              \path[->] (q_0) edge [loop right] node {a} ();
      \end{tikzpicture}
    \end{center}
    Both have an empty language ($\lang(q_0) = \lang(r_0) = \emptyset$),
    and in principle, we would like to be able to prove that $\gtripleDelay{\emptyset}{q_0}{\relL}{r_0}$.
    Unfortunately, since $q_0$ is accepting, the only rule that can be applied initially is \textsc{L-to-R} and we have to prove $\gtripleDelay{\emptyset}{q_0}{\relR}{r_0}$.
    This of course not possible as the right automaton does not have any accepting state.
  \end{example}

  This example is somewhat artificial. Indeed, since the left automaton does not have any word anyway, the accepting state $q_0$ is spurious and  can be removed.
  Without it, we would have $(q_0, r_0) \in \fsimdelay$, as expected.
  Unfortunately, this bad pattern can happen in practice 
  with accepting states that cannot be removed because they belong to a cycle.

  \begin{example}
    \label{ex:ddelay}
    Consider the following two automata:
  
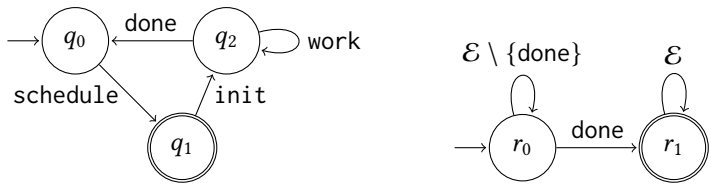
\begin{center}
    \begin{tikzpicture}[shorten >=1pt,node distance=2cm,on grid,auto,initial text=]
      \node[state,initial]
        (q_0) {$q_0$};
      \node[state,accepting] (q_1) [below right=of q_0] {$q_1$};
      \node[state] (q_2) [right=of q_0] {$q_2$};
      \path[->]
        (q_0) edge [left] node {\texttt{schedule}} (q_1)
        (q_1) edge [right] node {\texttt{init}} (q_2)
        (q_2) edge [above] node {\texttt{done}} (q_0)
        (q_2) edge [loop right] node {\texttt{work}} ();
    \end{tikzpicture}
    $\quad \quad$
    \begin{tikzpicture}[shorten >=1pt,node distance=2cm,on grid,auto,initial text=]
      \node[state,initial]    (q_0) {$r_0$};
      \node[state,accepting]  (q_1) [right=of q_0] {$r_1$};
      \path[->]
        (q_0) edge node {\texttt{done}} (q_1)
        (q_0) edge [loop above] node {$\mathcal{E} \setminus \{ \texttt{done} \}$} ()
        (q_1) edge [loop above] node {$\mathcal{E}$} ();
    \end{tikzpicture}
  \end{center}

  The left automaton could model a scheduler which controls the execution of a program.
  The accepting state $q_1$ models two assumptions: the execution of the program is scheduled infinitely many times, and once the program execution is scheduled, it always terminates and hands the control back to the scheduler. With this interpretation in mind, the right automaton specifies that the program should be executed until completion at least once (i.e., the transition \texttt{done} needs to be taken).
  We note that here, the fact that $q_1$ is accepting is crucial.

  It is easy to see that the left automaton is language included in the right one. Unfortunately, this fact cannot be established by delay simulation. Indeed, suppose we wish to prove $\gtripleDelay{\emptyset}{q_0}{\relL}{r_0}$. Since $q_0 \notin \mathcal{F}_1$, the best we can do is to use $\textsc{L-Step}$
  to go to $\gtripleDelay{\emptyset}{q_1}{\relL}{r_0}$. Now, since $q_1 \in \mathcal{F}_1$, we have no other choice than using $\textsc{L-to-R}$ and it remains to prove $\gtripleDelay{\emptyset}{q_1}{\relR}{r_0}$.
  From there, we have to show that no matter which path out of $q_1$ is taken, we are guaranteed to reach a right-accepting state from $r_0$. Unfortunately, this is not the case!
  Indeed, we have to first go to $\gtripleDelay{\emptyset}{q_2}{\relR}{r_0}$ by \textsc{R-Delay}, but since there is a self-loop $\iostep{q_2}{\texttt{work}}{q_2}$, the right-hand automaton is stuck in the non-accepting state $r_0$.
  \end{example}

  Even though the previous example cannot be handled by $\fsimdelay$ because $(q_2, r_0) \notin \fsimdelayR$, we observe that $(q_2, r_0) \in \fsimdelayL$! Further, since the B\"uchi acceptance condition requires executions to visit infinitely many times an accepting state, the first $n \in \mathbb{N}$ visits to a left-accepting state can always be ignored.
  Said otherwise, for this specific example, it would be correct to replace the left-hand automaton with the following one, where the two first steps have been explicitly unrolled, and the first visit to an accepting state is ignored: 
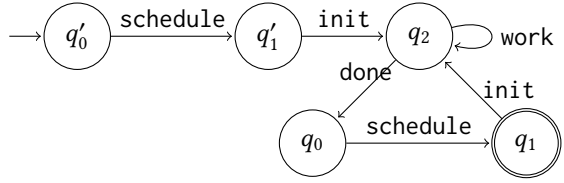
\begin{center}
    \begin{tikzpicture}[shorten >=1pt,node distance=2cm,on grid,auto,initial text=]
      \node[state,initial] (qU_0) {$q'_0$};
      \node[state] (qU_1) [right=of qU_0, xshift=15pt] {$q'_1$};
      \node[state] (q_2) [right=of qU_1] {$q_2$};
      \node[state] (q_0) [below left=of q_2] {$q_0$};
      \node[state,accepting] (q_1) [below right=of q_2] {$q_1$};
      \path[->]
	      (qU_0) edge node {\texttt{schedule}} (qU_1)
	      (qU_1) edge node {\texttt{init}} (q_2)
	      (q_0) edge [above] node {\texttt{schedule}} (q_1)
	      (q_1) edge [right] node {\texttt{init}} (q_2)
	      (q_2) edge [above] node {\texttt{done}} (q_0)
	      (q_2) edge [loop right] node {\texttt{work}} ();
    \end{tikzpicture}
  \end{center}

  Clearly, the unrolled version is language equivalent to the initial one. However,
  we can now prove $\gtriple{\emptyset}{q'_0}{\relL}{r_0}$. It suffices to use two times $\textsc{L-Step}$ to move to $\gtriple{\emptyset}{q_2}{\relL}{r_0}$.
  Since we have not encountered any left-accepting state on the way, we are still focusing on a $\texttt{L}$-triple. From there, it is not too difficult to conclude by using $\{ (q_2, r_0), (q_0, r_1), (q_1, r_1), (q_2, r_1) \}$ as a coinduction hypothesis.
  Instead of having to explicitly modify the left transition system, the purpose of this section is to develop several new notions of delay simulation that have built-in support to emulate this type of reasoning.

  \subsection{Double Delay Simulation}

  In the previous example, we encountered the problem that delay simulation tracks \emph{every} left-accepting state, and each time one is encountered, it forces us to prove that a right-accepting state can be later reached.
  Instead, we would like to be able to temporarily ignore left-accepting state, and \emph{choose} when to start matching them with right-accepting states.
  So long as we only allow to ignore left-accepting states for a bounded number of computation steps, this still gives a sound proof technique for language inclusion.

  A first idea to achieve this intuition is to prefix the greatest fixed point underlying the definition of $\fsimdelay$ with an extra least fixed point allowing to skip the first $n$ left-accepting states. We call the corresponding relation \emph{double delay simulation} and define it as follows:

  \begin{definition}[Double Delay Simulation]
    \begin{align*}
      \delay(X) &\triangleq \mu Y .\\
      \textsc{\color{gray}(Wait)} & \quad \{ \ (s_1, s_2) \mid \forall e. \forall \iostep{s_1}{e}{s'_1}. \exists \iostep{s_2}{e}{s'_2}. \ (s'_1, s'_2) \in Y \ \} \ \cup\\
      \textsc{\color{gray}(Commit)} & \quad \{ \ (s_1, s_2) \mid (s_1, s_2) \in X \ \}\\
      \ddelaysim &\triangleq \delay(\fsimdelayL)
    \end{align*}
  \end{definition}

  Concretely, to prove that $(s_1, s_2) \in \ddelaysim$ we can either decide to \emph{commit} to track left-accepting state right away by proving $(s_1, s_2) \in \fsimdelay$.
  Alternatively, we can also decide to \emph{wait} for one computation step, and instead show that every successor $s'_1$ of $s_1$ can be matched with a successor $s'_2$ of $s_2$ such that $(s'_1, s'_2)$ is again in $\ddelaysim$.
  Since $\texttt{delay}$ is an inductive predicate, we can only wait for finitely many successive computation steps before having to switch to $\fsimdelay$, which guarantees
  soundness for language inclusion.

  \begin{theorem}[Soundness]
    $(s_1, s_2) \in \ddelaysim \implies \lang(s_1) \subseteq \lang(s_2)$
  \end{theorem}
  \begin{proof}
    By induction on the definition of $\texttt{wait}$, exploiting the soundness of $\fsimdelayL$ for the base case.
  \end{proof}

  We observe that $\ddelaysim$ is strictly weaker than $\fsimdelay$. In particular, it is weak enough to cover the previous example.

  \begin{theorem}
    $\ddelaysim$ is strictly weaker than $\fsimdelay$.
  \end{theorem}

  \begin{proof}
    The fact that $\fsimdelay \subseteq \ddelaysim$ immediately follows from the \textsc{Commit} disjunct in the definition of \texttt{delay}. It remains to show that there exists at least one pair of B\"uchi automata such that $\ddelaysim$ contains strictly more pairs of states than $\fsimdelay$. 
    For the two automata from \Cref{ex:ddelay}, we already discussed that $(q_0, r_0) \notin \fsimdelay$.
    However, we show that $(q_0, r_0) \in \ddelaysim$.
    We start by unfolding the underlying least fixed point $\delay$ twice (exploiting the first disjunct {\color{gray}(\textsc{Wait})}) to get to $(q_2, r_0) \in \ddelaysim$.
    We then \emph{commit} to prove that a right-accepting state can be reached by unfolding \texttt{wait} and exploiting the second disjunct {\color{gray}(\textsc{Commit})}.
    From there it suffices to show that $(q_2, r_0) \in \fsimdelayL$. This can be done using the rules of \Cref{fig:rules_esim}. \begin{align*}
      &\gtripleDelay{\emptyset}{q_2}{\relL}{r_0}\\
      \textrm{\color{gray} By \textsc{L-Invariant}}
      \impliedby&\gtripleDelay{(q_2, r_0)}{q_2}{\relL}{r_0}\\
      \textrm{\color{gray} By \textsc{L-Step}}
      \impliedby&\tripleDelay{(q_2, r_0)}{q_0}{\relL}{r_1} \wedge \tripleDelay{(q_2, r_0)}{q_2}{\relL}{r_0}
    \end{align*}

    By \textsc{L-Cycle}, the second conjunct is trivially discharged, and it only remains to prove $\gtripleDelay{(q_2, r_0)}{q_0}{\relL}{r_1}$. \begin{align*}
      &\gtripleDelay{(q_2, r_0)}{q_0}{\relL}{r_1}\\
      \textrm{\color{gray} By \textsc{L-Invariant}}
      \impliedby&\gtripleDelay{(q_2, r_0), (q_0, r_1)}{q_0}{\relL}{r_1}\\
      \textrm{\color{gray} By \textsc{L-Step} and $q_0 \notin \mathcal{F}_1$}
      \impliedby&\tripleDelay{(q_2, r_0), (q_0, r_1)}{q_1}{\relL}{r_1}\\
      \textrm{\color{gray} By \textsc{L-Invariant}}
      \impliedby&\gtripleDelay{(q_2, r_0), (q_0, r_1), (q_1, r_1)}{q_1}{\relL}{r_1}\\
      \textrm{\color{gray} By \textsc{L-to-R} and \textsc{R-Delay}}
      \impliedby&\gtripleDelay{(q_2, r_0), (q_0, r_1), (q_1, r_1)}{q_2}{\relR}{r_1}\\
      \textrm{\color{gray} By \textsc{R-Final} and $r_1 \in \mathcal{F}_2$}
      \impliedby&\tripleDelay{(q_2, r_0), (q_0, r_1), (q_1, r_1)}{q_2}{\relL}{r_1} \ \wedge\\
      &\tripleDelay{(q_2, r_0), (q_0, r_1), (q_1, r_1)}{q_0}{\relL}{r_1}
    \end{align*}
    By \textsc{L-Cycle}, the right conjunct is discharged and it remains to show \begin{align*}
      &\gtripleDelay{(q_2, r_0),(q_0, r_1), (q_1, r_1)}{q_2}{\relL}{r_1}\\
      \textrm{\color{gray} By \textsc{L-Invariant}}
      \impliedby&\gtripleDelay{(q_2, r_0),(q_0, r_1), (q_1, r_1), (q_2, r_1)}{q_2}{\relL}{r_1}\\
      \textrm{\color{gray} By \textsc{L-Step} and $q_2 \notin \mathcal{F}_1$}
      \impliedby&\tripleDelay{(q_2, r_0), (q_0, r_1), (q_1, r_1), (q_2, r_1)}{q_2}{\relL}{r_1} \ \wedge\\
      &\tripleDelay{(q_2, r_0), (q_0, r_1), (q_1, r_1), (q_2, r_1)}{q_0}{\relL}{r_1}
    \end{align*}
    And by $\textsc{L-Cycle}$ we are done since both $(q_2, r_1)$ and $(q_0, r_1)$ are part of the coinduction hypothesis.
  \end{proof}

  \subsection{Repeated Delay}

  With $\ddelaysim$, left-accepting states can only be ignored at the beginning of a proof. Once we commit to track left-accepting states, we cannot use \textsc{Wait} anymore, and every subsequent visit to a left-accepting state will force us to prove that a right-accepting state can be reached.
  For the previous example, this was not an issue as it was enough to unroll the first two steps of the left automaton, skipping unnecessary $\mathcal{F}_1$-states until we are positioned right before the last event allowing us to conclude that our liveness goal is fulfilled.
  In general, this style of reasoning can be applied for reachability/termination specifications where, once an objective has been attained once, the left-automata is unconstrained.
  For richer specifications, this strategy does not always work.

  \begin{example}[Double Delay Simulation is still too strong]
    \label{ex:ddelay-too-strong}
    Consider the following example. The left automaton is the same as in \Cref{ex:ddelay}, but instead of requiring for the event \texttt{done} to be produced at least once, we require to be \texttt{done} infinitely many times.
    Concretely, the specification automaton now has a back edge forcing to transition from $r_1$ back to $r_0$.

    
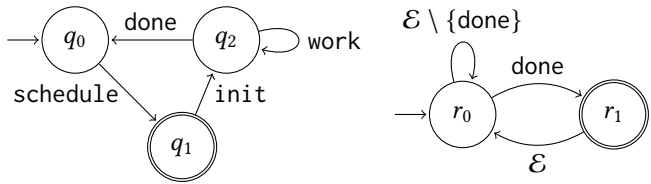
\begin{center}
      \begin{tikzpicture}[shorten >=1pt,node distance=2cm,on grid,auto,initial text=]
        \node[state,initial] (q_0) {$q_0$};
        \node[state,accepting] (q_1) [below right=of q_0] {$q_1$};
        \node[state] (q_2) [right=of q_0] {$q_2$};
        \path[->] (q_0) edge [left] node {\texttt{schedule}} (q_1)
          (q_1) edge [right] node {\texttt{init}} (q_2)
          (q_2) edge [above] node {\texttt{done}} (q_0)
          (q_2) edge [loop right] node {\texttt{work}} ();
      \end{tikzpicture}
      \begin{tikzpicture}[shorten >=1pt,node distance=2cm,on grid,auto,initial text=]
          \node[state,initial]  (q_0)                      {$r_0$};
          \node[state,accepting]          (q_1) [right=of q_0] {$r_1$};
    \path[->] (q_0) edge  [bend left]   node        {\texttt{done}} (q_1)
      (q_0) edge [loop above] node        {$\mathcal{E}\setminus\{\texttt{done}\}$} ()
      (q_1) edge [bend left] node {$\mathcal{E}$} (q_0);
      \end{tikzpicture}
    \end{center}

    With this modification of the specification, the right automaton does not double delay simulate the left one (i.e., $(q_0, r_0) \notin \ddelaysim$). We can use the double delay trick once to reach $(q_0, r_1)$, but then the specification automaton cycles back to $r_0$ and forces us to prove $(q_1, r_0) \in \fsimdelayL$. At this point, we cannot use \textsc{Wait} anymore, and since $q_1 \in \mathcal{F}_1$, we have to prove that $r_1$ can be reached again using only the rules of $\fsimdelayR$. As discussed in the previous section, this is not possible because of the self loop in $q_2$.
  \end{example}

  Fortunately, there is \textit{a priori} no reason to restrict ourselves to using the \texttt{wait} trick only at the beginning of a proof. Instead, every time a right-accepting state has been visited, it is correct to again temporarily ignore further visits to a left-accepting state.
  To concretize this intuition, we define the following \emph{repeated delay simulation} (we note $\texttt{fsim}_\texttt{rdelay}$). 
  
  \begin{definition}[Repeated Delay Simulation]
    \begin{align*}
      \rdelaysimR(X) &\triangleq \mu Y .\\
      \textsc{\color{gray}(R-Final)} & \quad \{ \ (s_1, s_2) \mid s_2 \in \mathcal{F}_2 \wedge \forall e. \forall \iostep{s_1}{e}{s'_1}. \exists \iostep{s_2}{e}{s'_2}. \ (s'_1, s'_2) \in \rdelaysimW(X) \ \} \ \cup\\
      \textsc{\color{gray}(R-Delay)} & \quad \{ \ (s_1, s_2) \mid \forall e. \forall \iostep{s_1}{e}{s'_1}. \exists \iostep{s_2}{e}{s'_2}. \ (s'_1, s'_2) \in Y \ \}\\
      \\
      \rdelaysimL &\triangleq \nu X .\\
      \textsc{\color{gray}(L-to-R)} & \quad \{ \ (s_1, s_2) \mid (s_1, s_2) \in \rdelaysimR(X) \ \} \ \cup\\
      \textsc{\color{gray}(L-Step)} & \quad \{ \ (s_1, s_2) \mid s_1 \notin \mathcal{F}_1 \wedge \forall e. \forall \iostep{s_1}{e}{s'_1}. \exists \iostep{s_2}{e}{s'_2}. \ (s'_1, s'_2) \in X \ \}\\
      \\
      \rdelaysimW(X) &\triangleq \mu Y.\\
      \textsc{\color{gray}(W-Wait)} & \quad \{ \ (s_1, s_2) \mid \forall e. \forall \iostep{s_1}{e}{s'_1}. \exists \iostep{s_2}{e}{s'_2}. \ (s'_1, s'_2) \in Y \ \} \ \cup\\
      \textsc{\color{gray}(W-Commit)} & \quad \{ \ (s_1, s_2) \mid (s_1, s_2) \in X \ \}\\
      \\
      \rdelaysim &\triangleq \rdelaysimW(\rdelaysimL)
    \end{align*}
  \end{definition}

  The intuition behind this definition is the following: Initially, we work with $\rdelaysimW$, which gives access to the \emph{unrolling rules} \textsc{W-Wait} and \textsc{W-Commit}, allowing us to ignore finitely many left-accepting states (same as $\delay$). Whenever we desire, we can commit to track left-accepting state and instead switch to the coinductive relation $\rdelaysimL$ (similar to $\fsimdelayL$). As soon as a left-accepting state is encountered while using $\rdelaysimL$, the control is transferred to the inductive relation $\rdelaysimR$ (similar to $\fsimdelayR$), which forces us to reach a right-accepting state in finitely many steps.
  Once a right-accepting state is found, instead of cycling back to $\rdelaysimL$ (as in $\fsimdelay$), we instead cycle back to $\rdelaysimW$.

  \begin{theorem}[Soundness]
    \label{thm:wsim_sound}
      $(s_1, s_2) \in \rdelaysim \implies \lang(s_1) \subseteq \lang(s_2)$
  \end{theorem}
  \begin{proof}(Sketch)
    The proof proceeds by coinduction on $\lang(s_2)$ and by induction on the inductive part of $\rdelaysimW$, the inductive part of $\lang(s_1)$ and the inductive definition of $\rdelaysimR$ (in this precise order).
    The technical details are rather tedious, and we refer the curious readers to our Rocq development for a complete proof.
  \end{proof}

  As for all the simulation relations presented so far, $\rdelaysim$ can be presented as a deductive system. We omit the exact definition of the triples associated with $\rdelaysim$. Instead, we directly present (a selection) of reasoning rules in \Cref{fig:rules_wsim}.

  \begin{figure}[ht!]
    \begin{mathpar}
      \inferrule[L-to-R]{\gtripleRdelay{H}{s_1}{\relR}{s_2}}{\gtripleRdelay{H}{s_1}{\relL}{s_2}}\qquad
      \inferrule[L-Step]{s_1 \notin \mathcal{F}_1\\\forall e. \forall \iostep{s_1}{e}{s'_1}. \exists \iostep{s_2}{e}{s'_2}. \ \tripleRdelay{H}{s'_1}{\relL}{s'_2}}{\gtripleRdelay{H}{s_1}{\relL}{s_2}}\\
      \inferrule[R-Delay]{\forall e. \forall \iostep{s_1}{e}{s'_1}. \exists \iostep{s_2}{e}{s'_2}. \ \gtripleRdelay{H}{s'_1}{\relR}{s'_2}}{\gtripleRdelay{H}{s_1}{\relR}{s_2}}\qquad
      \inferrule[R-Final]{s_2 \in \mathcal{F}_2\\\forall e. \forall \iostep{s_1}{e}{s'_1}. \exists \iostep{s_2}{e}{s'_2}. \ \gtripleRdelay{H}{s'_1}{\relW}{s'_2}}{\gtripleRdelay{H}{s_1}{\relR}{s_2}}\\
      \inferrule[W-Wait]{\forall e. \forall \iostep{s_1}{e}{s'_1}. \exists \iostep{s_2}{e}{s'_2}. \ \gtripleRdelay{H}{s'_1}{\relW}{s'_2}}{\gtripleRdelay{H}{s_1}{\relW}{s_2}} \qquad
      \inferrule[W-Commit]{\tripleRdelay{H}{s_1}{\relL}{s_2}}{\gtripleRdelay{H}{s_1}{\relW}{s_2}} \\
      \inferrule[L-Cycle]{(s_1, s_2) \in H}{\tripleRdelay{H}{s_1}{\relL}{s_2}} \qquad
      \inferrule[L-Guard]{\gtripleRdelay{H}{s_1}{\relL}{s_2}}{\tripleRdelay{H}{s_1}{\relL}{s_2}} \qquad
      \inferrule[L-Invariant]{(s_1, s_2) \in H'\\\forall (s'_1, s'_2) \in H'. \ \gtripleRdelay{H \cup H'}{s'_1}{\relL}{s'_2}}{\gtripleRdelay{H}{s_1}{\relL}{s_2}}
    \end{mathpar}
    \caption{Selection of proof rules for $\rdelaysim$}
    \label{fig:rules_wsim}
    \Description{Rules}
  \end{figure}

  We showcase the use of $\rdelaysim$ by applying the rules of \Cref{fig:rules_wsim} to the example of the beginning of this section (which was neither covered by $\fsimdelay$ nor by $\ddelaysim$).

  \begin{example}
	  \label{ex:rdelay_work}
	  Recall the following two automata from \Cref{ex:ddelay-too-strong}:
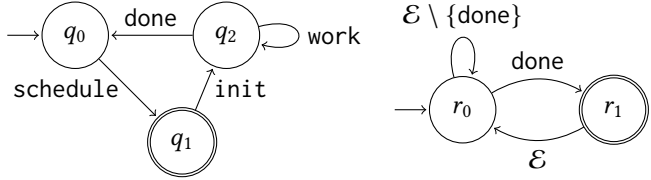
  \begin{center}
    \begin{tikzpicture}[shorten >=1pt,node distance=2cm,on grid,auto,initial text=]
      \node[state,initial] (q_0) {$q_0$};
      \node[state,accepting] (q_1) [below right=of q_0] {$q_1$};
      \node[state] (q_2) [right=of q_0] {$q_2$};
      \path[->] (q_0) edge [left] node {\texttt{schedule}} (q_1)
	      (q_1) edge [right] node {\texttt{init}} (q_2)
	      (q_2) edge [above] node {\texttt{done}} (q_0)
	      (q_2) edge [loop right] node {\texttt{work}} ();
    \end{tikzpicture}
    \begin{tikzpicture}[shorten >=1pt,node distance=2cm,on grid,auto,initial text=]
        \node[state,initial]  (q_0)                      {$r_0$};
        \node[state,accepting]          (q_1) [right=of q_0] {$r_1$};
	\path[->] (q_0) edge  [bend left]   node        {\texttt{done}} (q_1)
		(q_0) edge [loop above] node        {$\mathcal{E}\setminus\{\texttt{done}\}$} ()
		(q_1) edge [bend left] node {$\mathcal{E}$} (q_0);
    \end{tikzpicture}
  \end{center}
  We want to show $\gtripleRdelay{\emptyset}{q_0}{\relW}{r_0}$.
  We start by applying the \textsc{W-Wait} rule twice.
	\begin{align*}
		&\gtripleRdelay{\emptyset}{q_0}{\relW}{r_0} \\
		\textrm{\color{gray} By \textsc{W-Wait}} \impliedby&\gtripleRdelay{\emptyset}{q_1}{\relW}{r_0} \\
		\textrm{\color{gray} By \textsc{W-Wait}} \impliedby&\gtripleRdelay{\emptyset}{q_2}{\relW}{r_0} 
	\end{align*}
	Then, we are in $(q_2, r_0)$. From there, we commit and accumulate.
	\begin{align*}
		\textrm{\color{gray} By \textsc{W-Commit} and \textsc{L-Guard}} \impliedby&\gtripleRdelay{\emptyset}{q_2}{\relL}{r_0} \\
		\textrm{\color{gray} By \textsc{L-Invariant}} \impliedby&\gtripleRdelay{(q_2,r_0)}{q_2}{\relL}{r_0} 
	\end{align*}
	Since $q_2$ is not final, we can apply the \textsc{L-Step} rule. However, because $q_2$ has two outgoing transitions, our goal splits into two goals.
	The first transition stays in $(q_2, r_0)$ and thus we can conclude directly.
	\begin{align*}
		\textrm{\color{gray} By \textsc{L-Step}} \impliedby&\tripleRdelay{(q_2,r_0)}{q_2}{\relL}{r_0} \\
		\textrm{\color{gray} By \textsc{L-Cycle}} \impliedby&(q_2,r_0) \in \{(q_2,r_0)\} 
	\end{align*}
	The second transition goes to $(q_0, r_1)$. As $r_1$ is final, we can apply the \textsc{L-to-R} rule followed by \textsc{R-Final} to conclude.

	\begin{align*}
		\textrm{\color{gray} By \textsc{L-Step} and \textsc{L-Guard}} \impliedby&\gtripleRdelay{(q_2,r_0)}{q_0}{\relL}{r_1} \\
		\textrm{\color{gray} By \textsc{L-to-R}} \impliedby&\gtripleRdelay{(q_2,r_0)}{q_0}{\relR}{r_1} \\
		\textrm{\color{gray} By \textsc{R-Final}} \impliedby&\gtripleRdelay{(q_2,r_0)}{q_1}{\relW}{r_0} 
	\end{align*}
	Thus, we are now in $(q_1, r_0)$ and back in the wait relation. 
	We apply once the \textsc{W-Wait} rule to get to $(q_2, r_0)$. 
	Then, we commit and can directly conclude as we already have this state in our relation. 
	\begin{align*}
		\textrm{\color{gray} By \textsc{W-Wait}} \impliedby&\gtripleRdelay{(q_2,r_0)}{q_2}{\relW}{r_0} \\
		\textrm{\color{gray} By \textsc{W-Commit}} \impliedby&\tripleRdelay{(q_2,r_0)}{q_2}{\relL}{r_0} \\ 
		\textrm{\color{gray} By \textsc{L-Cycle}} \impliedby&(q_2,r_0) \in \{(q_2,r_0)\} 
	\end{align*}
  \end{example}

  \section{Related Work}

  \subsection{Fairness Preserving Simulations for B\"uchi Automata}

  There is a large body of literature
  studying fairness-preserving notions of simulation for B\"uchi automata \cite{mc-modular-verif, equivalence-fair-kripke,fair-simulation, delay-simulation, more-simulation-games}.
  \citeauthor{mc-modular-verif} introduced a first notion of simulation for transition systems augmented with a fairness condition \cite{mc-modular-verif}. To be precise, their definition covers transition systems equipped with a Street condition (note that such transition systems can be translated to B\"uchi automata, and vice versa).
  \citeauthor{equivalence-fair-kripke}~\cite{equivalence-fair-kripke} later extended the notion of simulation presented by \citeauthor{mc-modular-verif}~\cite{mc-modular-verif} into several alternative notions of bisimulation for transition systems equipped with a Muller fairness condition (again, these can be translated to B\"uchi automata, and vice versa).
  However, as later observed by \citeauthor{fair-simulation}~\cite{fair-simulation}, the definitions of \citeauthor{mc-modular-verif}~\cite{mc-modular-verif} and \citeauthor{equivalence-fair-kripke}~\cite{equivalence-fair-kripke} are not \emph{local}: both refer to fragments of executions instead of only referring to states and transitions. Consequently, these notions of fairness-preserving simulation are computationally expensive to automate, and they are also not well-suited for interactive proofs.
  \citeauthor{fair-simulation}~\cite{fair-simulation} introduced \emph{Fair Simulation} to overcome some of these limitations.
  Contrary to the definitions of \citeauthor{mc-modular-verif}~\cite{mc-modular-verif} and \citeauthor{equivalence-fair-kripke}~\cite{equivalence-fair-kripke}, the notion of fair simulation introduced by \citeauthor{fair-simulation}~\cite{fair-simulation} is more local. In particular, there are efficient algorithms to automatically prove fair simulation. Nonetheless, applications of fair simulation to interactive proofs are not investigated by \citeauthor{fair-simulation}~\cite{fair-simulation}.
  Further, to the best of our knowledge, none of the simulation techniques listed above have been mechanized in an interactive proof assistant.

  \subsection{Interactive Proofs using Simulation Relations}

  Over the course of the last decade, a variety of techniques and tools to integrate simulation techniques into interactive proof assistants have been developed.
  For example, our paper builds on the framework of \emph{parameterized coinduction} introduced by \citeauthor{paco}~\cite{paco} as a systematic methodology to develop interactive proofs by coinduction.
  Parameterized coinduction has been successfully applied to a wide range of problems
  including behavioral inclusion of reactive and impure programs modeled as interaction trees \cite{itrees}. It has also been used to develop powerful deductive systems to reason about weak notions of stream equivalence \cite{gpaco}.

  Another notable application of simulation techniques to interactive verification in a proof assistant is the Simuliris framework \cite{simuliris}. Simuliris combines a simulation relation with the mechanized concurrent separation logic Iris \cite{iris} in order to verify the correctness of optimizations for concurrent programs.
  At its core, Simuliris features a \emph{fair termination preserving} simulation relation: it ensures that terminating behaviors of the left program can be reproduced by the right program, and it further allows the left program to have diverging behaviors, so long as the right program can mimic these diverging behaviors via a fair execution. However, Simuliris cannot exploit fairness assumptions on the left program, nor does it support reasoning about more general liveness properties (other than termination/divergence).

  The recently developed notion of \emph{freely stuttering simulation} (FreeSim) \cite{for-free} is also related to our work.
  The core idea of FreeSim is to simplify interactive proofs of behavioral inclusion by simulation in the presence of \emph{stutter steps} (i.e., computation steps that do not produce events).
  In this setting, the right-hand system is typically allowed to stutter for finitely many computation steps until, eventually, synchronous progress is made and both the left-hand and the right-hand systems are emitting the same event (at which point the coinduction hypothesis can be used to close cycles).
  FreeSim proposes a systematic approach to further weaken this principle by instead allowing for \emph{asynchronous} progress to be made, effectively enabling (under additional constraints) to exploit the coinduction hypothesis after a stutter step in the right-hand system.
  Techniques developed by \citeauthor{for-free}~\cite{for-free} are similar in spirit to the techniques we employ in our notions of fairness-preserving simulation with repeated delay.
  We note, however, that FreeSim does not handle behavioral inclusions in the presence of fairness assumptions/requirements.

  A more recent approach that is connected to ours is the notion of \emph{Fair Operational Semantics} (FOS) introduced by \citeauthor{fos}~\cite{fos}.
  In essence, FOS proposes to model fairness assumptions of (concurrent) programs as an operational semantics. This is achieved by labelling the operational semantics of a programming language with specific events. As part of this framework, the authors show that it is possible to develop notions of simulation to prove behavioral inclusions in the presence of fairness assumptions modeled by FOS \cite{fos}.
  We note however that, as stated by the authors themselves, the notion of fairness supported by FOS differs in several subtle aspects with other notions of fairness that are commonly used in the automata-based model checking literature \cite{fos}. In particular, it is not yet clear whether arbitrary liveness specifications (and in particular, arbitrary LTL specifications) can easily be encoded using
  FOS. In contrast, the notions of simulation we developed in this paper apply to arbitrary B\"uchi automata, and by extension to any omega-regular specification (including in particular, LTL-definable specifications).

  \section{Conclusion and Future Work}

  In this paper, we re-investigated, from the point of view of interactive proofs, several notions of fairness-preserving simulations that are typically employed in the context of automata-based model-checking.
  Starting from \emph{direct simulation}, arguably the simplest notion of fairness-preserving simulation, we progressively constructed several improvements allowing to handle increasingly complex examples.
  Along the way, we demonstrated that by combining carefully designed fairness-preserving simulation relations with parameterized coinduction, we can obtain simple deductive systems to prove fair trace inclusion between transition systems modelling programs and B\"uchi automata modelling temporal specifications.
  We believe that this idea deserves further investigation. In particular, it could be used as a basis to develop a fully-featured mechanized program logic for temporal reasoning.

  \section*{Data-Availability Statement}

  All notions of simulation presented throughout the paper, their proofs of soundness, as well as all examples and counterexamples have been formalized in the Rocq proof assistant.
  The sources and their documentation are available at the following address \cite{artifact}:
  \begin{center}
    \url{https://zenodo.org/records/15658188}
  \end{center}

  \section*{Acknowledgements}

  This work was supported by the European Research Council (ERC) Grant HYPER (No. 101055412). Views and opinions expressed are however those of the authors only and do not necessarily reflect those of the European Union or the European Research Council Executive Agency. Neither the European Union nor the granting authority can be held responsible for them. A.~Correnson and I.~Kuhn carried out this work as members of the Saarbr\"ucken Graduate School of Computer Science.

  \bibliographystyle{ACM-Reference-Format}
  \bibliography{references}

\end{document}